%% file: main.tex
\documentclass[3p,onecolumn,12pt,times,longtitle]{elsarticle}

\usepackage{amsmath}
\usepackage{amssymb}
\usepackage{graphicx}
\usepackage{float}
\usepackage{color}
\usepackage{tcolorbox}
\usepackage[utf8]{inputenc}
\usepackage[english]{babel}
\usepackage{amsthm}

\usepackage{tikz-feynman}
\usepackage{pgfplots}
\usepackage{xparse}
\usetikzlibrary{positioning,arrows,patterns}
\usetikzlibrary{decorations.markings}
\usetikzlibrary{calc}

\tikzset{
  photon/.style={decorate, decoration={snake}, draw=black},
  fermion/.style={draw=black, postaction={decorate},decoration={markings,mark=at position .55 with {\arrow{>}}}},
}

\input{latex_command.tex}

\journal{Nuclear Physics B}

\begin{document}

\begin{frontmatter}

\title{Non-BPS Supersymmetric 3pt Amplitude for \\One Massless, Two Equally Massive Particles}

\input{authorlist}

\AddAuthor{Bo-Ting Chen}{1}{}{}%
\AddInstitute{1}{Department of Physics and Astronomy, National Taiwan University, Taipei 10607, Taiwan}

\date{\today}

\begin{abstract}
In this paper, the non-BPS amplitudes ($Z<2m$) are considered. Utilizing on-shell methods, the three point amplitudes of two equal-mass particles and one massless particle were constructed, where the two massive particles are non-BPS states. We verify the result by matching the $Z=0$ and BPS limit with $\Math{N}=2$ supersymmetry. As an application we derive the non-BPS coupling for $\Math{N}=4$ super-Maxwell and supergravity.
\end{abstract}

\end{frontmatter}


\section{Introduction}

Spinor helicity formalism enables us to derive S-matrix by on-shell formulation. Unlike Feynman rules, the formulation does not introduce gauge redundancy into the computation. Rather, the amplitude can be fully determined by the momentum and the spin polarization of the external particles. Recently, spinor helicity formalism is adapted to describe four-dimensional scattering amplitudes for particles of any mass and spin~\cite{Arkani-Hamed:2017jhn}. A 3-pt amplitude of two equal massive particles with mass $m$ and a massless particle with helicity $h$ can be written in spinor helicity basis, see~\cite{Arkani-Hamed:2017jhn},
\ie
\label{eqn: amplitude without SUSY}
M^{(I_1\cdots I_{2s_1})(J_1\cdots J_{2s_2})h}=
\left(\lambda_1\right)^{I_1}_{\alpha_1}\cdots\left(\lambda_1\right)^{I_{s_1}}_{\alpha_{s_1}}
\left(\lambda_2\right)^{I_2}_{\beta_2}\cdots\left(\lambda_2\right)^{I_{s_2}}_{\beta_{s_2}}
M^{(\alpha_1\cdots \alpha_{2s_1})(\beta_1\cdots \beta_{2s_2})h}
\;.\fe

Supersymmetry (SUSY) requires on-shell fermionic variables. Spinor helicity formalism then must be formulated in on-shell superspace. The formulation in massless~\cite{Elvang:2015rqa} and massive~\cite{Boels:2011zz,Herderschee:2019ofc} on-shell superspace were subsequently constructed. Here we consider extended SUSY with non-vanishing central charge, using superamplitudes in $\Math{N}=2$ SUSY as building blocks. For extended $\Math{N}=2$ SUSY, the algebra takes the form
\ie
\begin{cases}
&\{ Q^A_{\alpha}, \tQ_{\dot{\alpha}B} \}=p_{\alpha\dot{\alpha}}\delta^A_B
\\
&\{ Q^A_{\alpha}, Q^B_{\beta} \}={1\over2} Z^{AB}\epsilon_{\alpha\beta}
\\
&\{ \tQ_{\dot{\alpha}A}, \tQ_{\dot{\beta}B} \}=-{1\over2} Z_{AB}\epsilon_{\dot{\alpha}\dot{\beta}}
\end{cases}
\;\;;\;\;
Z_{AB}=
\begin{bmatrix}
0&-Z\\
Z&0
\end{bmatrix}_{AB}
=Z\cdot\epsilon_{AB}
\;,\fe
and the central charge of a massive particle has a bound, $Z\le2m$. BPS states are the states that saturate this bound. The solution to the amplitude of BPS states are studied in~\cite{Cachazo:2018hqa,Caron-Huot:2018ape}.

\begin{figure}[H]
\label{fig: 3-pt vertex}
\centering
\begin{flalign}
\sMsasb{s_1}{s_2}{h}=\;\;\;\;\;\;\;\;\;\;\;\;\;\;\;\;\;\;\;\;\;\;\;\;\;\;\;\;\;\;\;\;\;\;\;\;\;\;\;\;\;\;\;\;\;\;\;\;
\end{flalign}
\begin{tikzpicture}[node distance=0.65cm and 0.65cm]
\coordinate[] (v1);
\coordinate[right=of v1] (v2);
\coordinate[right=of v2] (v3);
\coordinate[right= of v3] (v4);
\coordinate[above left = 1.75cm and 1.75cm of v1,label=left :$\bo^{(I_1\cdots I_{2s_1})}$ ] (e1);
\coordinate[below left = 1.75cm and 1.75cm of v1,label=left :$\bt^{(J_1\cdots J_{2s_2})}$ ] (e2);
\draw[fermion] (e1) -- (v1) node[midway, sloped, above=0.1cm,font=\footnotesize] {$p_1, (m,Z)$};
\draw[fermion] (e2) -- (v1) node[midway, sloped, below=0.1cm,font=\footnotesize] {$p_2, (m,-Z)$};
\draw[photon] (v1) -- (v4) node[midway,below=0.1cm] {$\longleftarrow$} node[right] {$h$} node[midway, above=0.1cm]{$P$};
\end{tikzpicture}
\caption{A 3-pt superamplitude of two massive multiplets with equal mass $m$ and a massless multiplet with helicity $h$. Note that the central charges of the two massive particles are opposite to each other ($Z$ and $-Z$), due to central charge conservation.}
\end{figure}
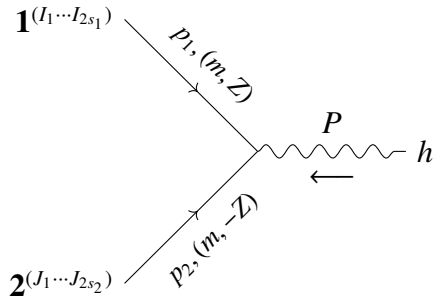

In this paper, we will consider the non-BPS three point superamplitude shown in Figure \ref{fig: 3-pt vertex}. The amplitude $\Math{M}_Z$ is made up of two massive multiplets, with spin-$2s_1$ and spin-$2s_2$, coupling a spin-$h$ massless multiplet, where the two massive multiplets have equal masses $m$ and opposite central charges ($Z$ and $-Z$, respectively)\footnote{To distinguish supersymmetric amplitudes from non-supersymmetric amplitudes, we denote supersymmetric amplitudes by curly alphabet $\Math{M}$.}. We require $\Math{M}_Z$ satisfy supersymmetry, i.e., 
\ie
Q_\alpha^A \Math{M}_Z=\tQ_{\dot{\alpha}}^{A} \Math{M}_Z=0
\;,\fe
where $Q_\alpha^A$ and $\tQ_{\dot{\alpha}}^{A}$ are the supersymmetric generators, and $A=1,2$ for $\Math{N}=2$ SUSY. By examining the the form of generators, we can show that $\Math{M}_Z$ is a function of $\delta^A$, $\Ha^A$, and $\Hb^A$, which are defined in \eqref{eqn: the building blocks}. In addition, we also show that we can always factorize the superamplitude with spinning multiplets in $\Math{N}=2$ SUSY into two parts,
\ie
\label{M and A}
&\sMsasb{s_1}{s_2}{h}=\Msasb{s_1}{s_2}{h}\cdot \AAA\left(Z,m\right)
\\
&Q_\alpha^A \AAA=0\;;\;\tQ_{\dot{\alpha}A} \AAA=0
\;.\fe
One of them is a bosonic factor $\Msasb{s_1}{s_2}{h}$, which carries the little group (LG) indices, including both the $SU(2)$ indicies for the massive and $U(1)$ for the massless multiplets. The other is a LG neutral and SUSY invariant quantity $\AAA\left(Z,m\right)$\footnote{Note that SUSY only constraints $\AAA$, but not the bosonic factor $\Msasb{s_1}{s_2}{h}$.}. Therefore, to solve $\Math{M}_Z$, we need to solve $\AAA$. There are 2 solutions for $\AAA$, and $\Math{M}_Z$ is a combination of the two solutions, see \eqref{eqn: spinning multiplet the most general form}.  We verify the solutions of $\AAA$ by matching to known results, i.e., the $Z=0$ limit~\cite{Herderschee:2019ofc} and the BPS limit~\cite{Caron-Huot:2018ape}.

As as application, we use our results in $\Math{N}=4$ SUSY. The presence central charge breaks R-symmetry from $SU(4)$ to $SU(2)\otimes SU(2)$, and the central charge is
\ie
\label{eqn: N=4 central charge-introduction}
Z_{AB}=
\begin{bmatrix}
0&-Z_{12}&0&0\\
Z_{12}&0&0&0\\
0&0&0&-Z_{34}\\
0&0&Z_{34}&0
\end{bmatrix}
\;,\fe
where $A,B=1,2,3,4$. We can then treat $\AAA$ as a "seed" for the superamplitude in $\Math{N}=4$ SUSY. As an example,
\ie
\Math{M}_{Z_{12},Z_{34}}(\bold{\Phi},\bold{\Phi},\phi)=\AAA_{12}\cdot\AAA_{34}
\;,\fe
where the subscripts $12$ or $34$ indicate which projected $SU(2)$ group they are describing. As in $\Math{N}=2$ SUSY, the superamplitude for spinning multiplets in $\Math{N}=4$ SUSY can be factorized as well, similar to \eqref{M and A} (see \eqref{eqn: spinning multiplet the most general form, N=4} for more details). We then proceed to consider super Maxwell theories and super-gravity (SUGRA) theories by choosing suitable bosonic factors. In the former case, the massless multiplet carries helicity $+1$, while in the latter case, there are 2 conjugated superamplitudes, and the massless multiplets in the superamplitudes carry helicity $2$ and helicity $0$.

In section \ref{section: The generators and the building blocks}, we show that $\Math{M}_Z$ is a function of $\delta^A$, $\Ha^A$, and $\Hb^A$ by examining the generators. In section \ref{section: N=2}, we solve for $\AAA$ and $\Math{M}_Z$ by requiring them being SUSY invariant. In section \ref{section: Z=0 and BPS and massless}, we compare our results in the $Z=0$ limit and the BPS limit with previous works. In section \ref{section: N=4 super-Maxwell and SUGRA}, we apply the results to $\Math{N}=4$ non-BPS SUSY, and explore super-Maxwell and super-gravity theories.

\section{The building blocks for the 3pt amplitude}
\label{section: The generators and the building blocks}

As mentioned in the introduction, we are dealing with the 3pt amplitude $\MmmP$ consisting of two equal-mass $m$, non-BPS particles (leg 1 and 2) and a massless particle (leg $P$), see Figure \ref{fig: 3-pt vertex}. In this chapter, we will introduce $\Math{N}=2$, non-BPS generators. Given the generators, we find the general form of $\MmmP$ is
\ie
\MmmP=\left[\prod_{A=1}^2\delta^A\right]\cdot\MmmPr\left(\Ha,\Hb\right)
\;,\fe
where
\ie
\label{eqn: the building blocks}
\delta^A&\equiv\sum_{i=1}^2\LRA{P\bi^I}\eta_{iI}^A\;\;,\;\;
\Ha^A &\equiv  \frac{1}{\LRA{\zeta P}} \sum_{i=1}^2\LRA{\zeta \bi^I}\eta_{iI}^A\;\;,\;\;
\Hb^A &\equiv  \frac{1}{\LRB{\xi P}} \sum_{i=1}^2\LRB{\xi \bi^I}\sigma_i\eta_{iI}^A
\;,\fe
are the building blocks for the 3pt amplitude ($\LA{\zeta}$ and $\LB{\xi}$ are reference spinors, and the dependence on the reference spinors will drop out if $\Ha$ and $\Hb$ is evaluated on the support of $\left[\prod_{A}\delta^A\right]$). In other words, the superamplitude is proportional to a delta function $\left[\prod_{A}\delta^A\right]$, and the rest of the amplitude is a function of $\Ha$ and $\Hb$.

\subsection{The generators}
The generators for massless particles satisfy the anti-commutation rules
\ie
\label{eqn: the anti-commutation relations- massless}
&\{ Q^A_{\alpha}, \tQ_{\dot{\alpha}B} \}=p_{\alpha\dot{\alpha}}\delta^A_B
\\
&\{ Q^A_{\alpha}, Q^B_{\beta} \}=0
\\
&\{ \tQ_{\dot{\alpha}A}, \tQ_{\dot{\beta}B} \}=0
\;,\fe
while for massive particles, the generators satisfy
\ie
\label{eqn: the anti-commutation relations- massive}
&\{ Q^A_{\alpha}, \tQ_{\dot{\alpha}B} \}=p_{\alpha\dot{\alpha}}\delta^A_B
\\
&\{ Q^A_{\alpha}, Q^B_{\beta} \}={1\over2} Z^{AB}\epsilon_{\alpha\beta}
\\
&\{ \tQ_{\dot{\alpha}A}, \tQ_{\dot{\beta}B} \}=-{1\over2} Z_{AB}\epsilon_{\dot{\alpha}\dot{\beta}}
\;,\fe
where $Z_{AB}=Z\cdot\LCT_{AB}$ is the central charge. Note that $Z$ has a bound, $0\le |Z|\le 2m$, and BPS states are states that satisfy $|Z|=2m$, where calculations will be lot simpler. However, we are interested in the amplitude for general $Z$.

In the 3pt amplitude we are considering, the central charges of the two massive particles must carry opposite signs, $Z_1+Z_2=0$, due to central charge conservation. We then define
\ie
\sin({2\tZ})\equiv\frac{Z_1}{2m}
\;.\fe
The bound $-\frac{\pi}{2}\le\tZ\le\frac{\pi}{2}$ follows directly from the bound of the central charge $0\le |Z|\le 2m$, and BPS limit happens at $\tZ=\frac{\pi}{2}$. For convenience, we also define $\CosT=\cos(\tZ)$ and $\SinT=\sin(\tZ)$.

To be a superamplitude, $\MmmP$ should satisfy
\ie
\label{eqn:QA and Q+A}
&Q_\alpha^A\MmmP=0\;;\;\tQ_{\dot{\alpha}A}\MmmP=0
\;,\fe
where $Q_\alpha^A$ and $\tQ_{\dot{\alpha}A}$ are the sum of the supersymmetric generators of the particles
\ie
\label{Q and Q+ sum}
Q_\alpha^A\equiv Q_{1\alpha}^A+Q_{2\alpha}^A+Q_{P\alpha}^A\;;\;\tQ_{\dot{\alpha}A}\equiv \tQ_{1\dot{\alpha}A}+\tQ_{2\dot{\alpha}A}+\tQ_{P\dot{\alpha}A}
\;.\fe
By introducing a set of Grassmann variables $\{\eta_{iI}^A,\eta_{P}^A\}\;\;(i=1,2)$\footnote{In this paper, the Grassmann variables' indices are defined as:\ie\begin{cases}i\text{: the label of the massive particles, }i=1,2\\I\text{: the LG indices for the massive particles}\\A\text{: the R-charge indices.}\end{cases}\;.\fe}, we can express the generators as
\ie
\begin{cases}
Q_{i\alpha}^A=\lambda_{i\alpha}^I\left(\CosT\eta_{iI}^A+\SinT\sigma_i\frac{\partial}{\partial\eta_{iA}^I}\right)
\\
\tQ_{i\dot{\alpha}A}=\tilde{\lambda}_{i\dot{\alpha}}^I\left(\CosT\frac{\partial}{\partial\eta_{i}^{IA}}-\SinT\sigma_i\eta_{iIA}\right)
\end{cases}
\;;\;
&\begin{cases}
Q_{P\alpha}^A=\lambda_{P\alpha}\eta_P^A
\\
\tQ_{P\dot{\alpha}A}=\tilde{\lambda}_{P\dot{\alpha}}\frac{\partial}{\partial\eta_P^A}
\end{cases}
\;,\fe
where $(\sigma_1,\sigma_2)=(1,-1)$. One can verify them by substituting them into \eqref{eqn: the anti-commutation relations- massless} and \eqref{eqn: the anti-commutation relations- massive}. Therefore, the sum of the generators are
\ie
\label{eqn: QQ+ full}
&Q_\alpha^A=\sum_{i=1}^2\lambda_{i\alpha}^I\left(\CosT\eta_{iI}^A+\SinT\sigma_i\frac{\partial}{\partial\eta_{iA}^I}\right)+\lambda_{P\alpha}\eta_P^A
\\
&\tQ_{\dot{\alpha}A}=\sum_{i=1}^2\tilde{\lambda}_{i\dot{\alpha}}^I\left(\CosT\frac{\partial}{\partial\eta_{i}^{IA}}-\SinT\sigma_i\eta_{iIA}\right)+\tilde{\lambda}_{P\dot{\alpha}}\frac{\partial}{\partial\eta_P^A}
\;.\fe

\subsection{The fermionic delta function}
\label{subsection:  the delta function}

From \eqref{eqn: QQ+ full}, we see that the generators can be separated into a multiplicative part and a differential part. By suitably combining the generators, we are able to subtract the differential part, and the superamplitude will be found to be proportional to the remaining multiplicative part. More explicitly, we can show that $\MmmP$ is proportional to the product of two fermionic delta functions
\ie
\MmmP\propto\prod_{A=1}^2\delta^A
\;,\fe
where $\delta^A\equiv\sum_{i=1}^2\LRA{P\bi^I}\eta_{iI}^A$.

Let's contract \eqref{eqn:QA and Q+A} with $\lambda_P^\alpha$ and $\tilde{\lambda}_{P\dot{\alpha}}$ to get rid of $\eta_P$'s, and multiply them with some proper coefficients, we get
\ie
&\CosT\LRA{PQ^A}\MmmP=0&\Rightarrow\sum_{i=1}^2\LRA{P\bi^I}\left(\CosTsqr\eta_{iI}^A+\SinT\CosT\sigma_i\frac{\partial}{\partial\eta_{iA}^I}\right)\MmmP=0
\;,\\
&{1\over x} \epsilon^{AB}\SinT\LRB{P\tQ_B}\MmmP=0&\Rightarrow\sum_{i=1}^2\LRA{P\bi^I}\left( \SinT\CosT\sigma_i\frac{\partial}{\partial\eta_{iA}^{I}}+\SinTsqr\eta_{iI}^A  \right)\MmmP=0
\;.\fe
By subtracting the two equations, we can see that if $\CosTsqr-\SinTsqr\neq0$, then
\ie
\sum_{i=1}^2\LRA{P\bi^I}\eta_{iI}^A\MmmP=0
\;.\fe
This result implies that $\MmmP$ is proportional to the delta function 
\ie
\label{eqn: the delta function}
\delta^{2}(\eta)=\prod_{A=1}^2\sum_{i=1}^2\LRA{P\bi^I}\eta_{iI}^A
\;.\fe
We then write $\MmmP$ as
\begin{tcolorbox}
\ie
\label{eqn: A=delta B}
\MmmP=\left[\prod_{A=1}^2\sum_{i=1}^2\LRA{P\bi^I}\eta_{iI}^A\right]\MmmPr
\;.\fe
\end{tcolorbox}

We now want to see what constraints SUSY places on $\MmmPr$. Since the calculations of $\MmmPr$ later in this section will be on the support of the delta function $\delta^{2}(\eta)$, we introduce a notation "$\eqd$", and denote $A$ is equal to $B$ on the support of the delta function as $A\eqd B$. In other words,
\ie
A\eqd B\;\Leftrightarrow\; \delta^{(2)}A=\delta^{(2)}B
\;.\fe
To see what constraints SUSY places on $\MmmPr$, first observe that when one imposes momentum conservation $p_1+p_2+P=0$, the following two commutators vanishes
\ie
\left[Q^A_{\alpha},\prod_{A=1}^2\sum_{i=1}^2\LRA{P\bi^I}\eta_{iI}^A\right]&=0 \text{, } \left[\tQ_{i\dot{\alpha}A},\prod_{A=1}^2\sum_{i=1}^2\LRA{P\bi^I}\eta_{iI}^A\right]&=0
\;.\fe
Therefore, according to \eqref{eqn:QA and Q+A}, $\MmmPr$ should satisfy
\ie
&\left[\prod_{A=1}^2\sum_{i=1}^2\LRA{P\bi^I}\eta_{iI}^A\right]Q^A_{\alpha}\MmmPr=0 \text{ ; } \left[\prod_{A=1}^2\sum_{i=1}^2\LRA{P\bi^I}\eta_{iI}^A\right]\tQ_{i\dot{\alpha}A}\MmmPr=0
\;.\fe
\noindent
Let's introduce two reference spinors $\LA{\zeta}$ and $\LB{\xi}$, which are not parallel to $\LA{P}$ and $\LB{P},$\footnote{In the rest of this paper, $\LA{\zeta}$ and $\LB{\xi}$ stand for reference spinors, and they are not parallel to $\LA{P}$ and $\LB{P}$.}
\ie
\LRA{\zeta P}\ne0\;;\;\LRB{\xi P}\ne0
\;.\fe
To simplify notations, let's further define
\ie
\label{eqn: definition of D and D+}
&\begin{cases}
\DaRaise{A}\equiv\sum_{i=1}^2 \LRA{\zeta \bi^I} \CosT\eta_{iI}^A
\\
\DaLower{A}\equiv\sum_{i=1}^2 \LRA{\zeta \bi^I} \SinT\sigma_i\frac{\partial}{\partial\eta_{iA}^I}
\end{cases}
\Rightarrow \Da{A}=\DaRaise{A}+\DaLower{A}
\\
&\begin{cases}
\DbRaise{A}\equiv\sum_{i=1}^2 \LRB{\xi \bi^I} \left(-\SinT\right)\sigma_i\eta_{iIA}
\\
\DbLower{A}\equiv\sum_{i=1}^2 \LRB{\xi \bi^I}  \CosT\frac{\partial}{\partial\eta_{i}^{IA}}
\end{cases}
\Rightarrow \Db{A}=\DbRaise{A}+\DbLower{A}
\;,\fe
and we have
\begin{tcolorbox}
\ie
\label{eqn:constraints on B}
&\Da{A}\MmmPr + \LRA{\zeta P}\eta_P^A\MmmPr \eqd0\\
&\Db{A}\MmmPr + \LRB{\xi P}\frac{\p}{\p\eta_P^A}\MmmPr \eqd0
\;.\fe
\end{tcolorbox}
\noindent
The plus and minus signs in the subscripts are indicating that they raise/lower the $\eta_i$-order\footnote{In this paper, the order of $\eta_i$ means the order of the sum of $\eta_1$ and $\eta_2$. For instance, $\eta_{1I}^A\eta_{2J}^B$ has order 2, $\eta_{1I}^A\eta_{2J}^B\eta_{1KA}$ has order 3. Note that $\eta_{1I}^A\eta_{2J}^B\eta_P^C$ has order 2, since $\eta_P^A$ doesn't increase the order of $\eta_i$.}.

\subsection{$\Ha$ and $\Hb$}
On the delta function, we have (simply because they are proportional to $\delta^A$)
\ie
\LA{P} \left[ \sum_{i=1}^2\RA{\bi^I}\eta_{iI}^A \right]\eqd0\;\;;\;\;\LB{P} \left[ \sum_{i=1}^2\RB{\bi^I}\sigma_i\eta_{iI}^A \right]\eqd0
\;,\fe
therefore,
\ie
\left[ \sum_{i=1}^2\RA{\bi^I}\eta_{iI}^A \right]\propto\RA{P} \;\;;\;\;\left[ \sum_{i=1}^2\RB{\bi^I}\sigma_i\eta_{iI}^A \right]\propto\RB{P} 
\;.\fe
Let's define the proportionality coefficient to be $\Ha$ and $\Hb$
\begin{tcolorbox}
\ie
\label{eqn: expressing via reference spinors}
\Ha^A &\equiv  \frac{1}{\LRA{\zeta P}} \sum_{i=1}^2\LRA{\zeta \bi^I}\eta_{iI}^A\\
\Hb^A &\equiv  \frac{1}{\LRB{\xi P}} \sum_{i=1}^2\LRB{\xi \bi^I}\sigma_i\eta_{iI}^A
\;.\fe
\end{tcolorbox}
\noindent
$\Ha$ and $\Hb$ are independent to those reference spinors when evaluated on the support of the delta function. For more details about $\Ha$ and $\Hb$, see \ref{appendix: More on Ha and Hb}.

In \ref{appendix: amplitude = f(Ha,Hb)}, we show that $\MmmPr$ is a function of $\Ha$ and $\Hb$ only. More precisely, all $\eta_i$ ($\eta_i=\;\eta_1\text{ or }\eta_2$) in $\MmmPr$ must be of the form
\ie
\sum_{i=1}^2\lambda_i ^I\eta_{iI}^A=\lambda_P\Ha^A  \;\;\text{or}\;\;  \sum_{i=1}^2\tilde\lambda_i ^I\sigma_i\eta_{iI}^A=\tilde\lambda_P\Hb^A
\;.\fe
$\delta^A$, $\Ha^A$, $\Hb^A$ play central roles in obtaining the superamplitude. In fact, if a superamplitude satisfies \eqref{eqn:QA and Q+A}, then all its $\eta_i$'s (massive Grassmann variables) will appear in the form of $\delta^A$, $\Ha^A$, $\Hb^A$ (proven in \ref{appendix: amplitude = f(Ha,Hb)}). Therefore, $\delta^A$, $\Ha^A$, $\Hb^A$ serve as "building blocks" of the superamplitude. To simplify notations, we also define
\ie
\HHa \equiv \frac{1}{2}\Ha^A\Ha_A \;\;;\;\; \HHb \equiv \frac{1}{2}\Hb^A\Hb_A \;\;;\;\; \HHab\equiv\frac{1}{2}\Ha^A\Hb_A
\;.\fe

\paragraph{Raising and lowering $\eta$ order}
$ $\newline
For an arbitrary quantity $X$, we can raise its $\eta_i$-order by acting $\mathcal{D}_{+\alpha}^A$ and $\tilde{\mathcal{D}}_{+\dot{\alpha} A}$ (see \eqref{eqn: definition of D and D+} for their definitions)
\ie
\label{eqn: raising eta order}
\DaRaise{A} X&=\LRA{\zeta P}\cdot \CosT\Ha^A X
\\
\DbRaise{A} X&=\LRB{\zeta P}\cdot\left(-\SinT\right)\Hb_A X
\;.\fe
\\
On the other hand, if a quantity meets $\mathcal{D}_{-\alpha}^A$ or $\tilde{\mathcal{D}}_{-\dot{\alpha} A}$, the $\eta_i$-order is lowered,
\ie
\begin{cases}
\label{eqn: lowering eta order}
\DaLower{A} \Ha ^B=0\\
\DbLower{A} \Ha ^B=\LRB{\zeta P}\cdot(-\CosT)\delta^B_A\\
\DaLower{A} \Hb _B=\LRA{\zeta P}\cdot \SinT\delta^A_B\\
\DbLower{A} \Hb _B=0
\end{cases}
\;.\fe


\section{Solutions to $\Math{N}=2$}
\label{section: N=2}
In this section, we will consider $\Math{N}=2$ super symmetry, and look for the explicit solution of $\AAA$ introduced in \eqref{M and A} (which is LG neutral and SUSY invariant), using the building blocks introduced in section \ref{section: The generators and the building blocks}. We will also consider spinning multiplets, and show that they can always be written in the form of \eqref{M and A}.

The massless super field with vacuum of helicity $h$ in $\Math{N}=2$ SUSY is
\ie
{\Psi^{+h}}=\field{h}+\eta^A\field{h-\tfrac{1}{2}}_A+\eta^A\eta_A\field{h-1}
\;.\fe
On the other hand, the scalar massive multiplet in $\Math{N}=2$ SUSY is \cite{Cachazo:2018hqa}
\ie
\bold{\Phi}=\phi+\eta_I^A\psi_A^I+\frac{1}{2}\eta_I^A\eta_{JA} W^{(IJ)}+\frac{1}{2}\eta_I^A\eta^{IB} \phi_{(AB)} +\frac{1}{3}\eta_I^A\eta_{JA}\eta^{JB}\bar\psi_B^I+\frac{1}{12}\eta^A_I\eta^{IB}\eta_{JA}\eta_B^J\bar\phi
\;.\fe
Note that in the above equations, $I,J$ stand for LG indices, while $A,B$ stand for R-charge indices. From the above expansion, we can see that there are 5 spin-$0$ components (1 $\phi$, 1 $\bar{\phi}$, and 3 $\phi_{(AB)}$), 4 spin-$1\over2$ components (2 $\psi_A^I$, and 2 $\bar{\psi}_A^I$), and 1 spin-$1$ component ($W^{(IJ)}$). Since a spin-$s$ particles carries $(2s+1)$ degree of freedom, our massive multiplet has 8 bosonic d.o.f and 8 fermionic d.o.f.

\subsection{The SUSY invariant quantity $\AAA$}
$\AAA$ is the quantity that satisfies
\begin{enumerate}
\item invariant under all LG transformations of the external particles
\item invariant under SUSY
\end{enumerate}
Since $\AAA$ is annihilated by the super charges, all the discussions about $\MmmP$ in section \ref{section: The generators and the building blocks} applies to $\AAA$ as well. According to \eqref{eqn: A=delta B},
\ie
\AAA=\delta^1\delta^2\AAAr
\;,\fe
where $\AAAr$ is a function of $\Ha$ and $\Hb$ (just like $\MmmPr$ is a function of $\Ha$ and $\Hb$). Let's expand $\AAAr$ in $\eta_P$
\ie
\label{eqn: B, f, g, h}
\AAAr\equiv  \Af+\eta_P^A \Ag_A+{1\over2}\eta_P^A\eta_{PA} \Ah
\;,\fe
where $\Af$, $\Ag$, $\Ah$ are all functions of $\Ha$ and $\Hb$. According to \eqref{eqn:constraints on B}, $\Af$, $\Ag$, and $\Ah$ should satisfy
\ie
\label{eqn: true constraints on f, g, h}
\begin{cases}
\mathcal{D}_\alpha^A \Ag_B-\delta^A_B\lambda_{P\alpha}\Af \eqd 0\\
\mathcal{D}_\alpha^A \Ah-\epsilon^{AB}\lambda_{P\alpha}\Ag_B \eqd 0\\
\tilde{\mathcal{D}}_{\dot{\alpha} A}\Ah \eqd 0
\end{cases}
\;.\fe
To get the full superamplitude, let's first solve $\Ah$, which should satisfy $\tilde{\mathcal{D}}_{\dot{\alpha} A}\Ah=0$. We can do this by first writing down all possible combinations of $\Ha$ and $\Hb$ at each $\eta_i$ order, and fix their coefficients by demanding that $\Ah$ satisfies $\tilde{\mathcal{D}}_{\dot{\alpha} A}\Ah=0$. After solving $\Ah$, we can obtain $\Ag_A$ and $\Af$ by resorting to \eqref{eqn: true constraints on f, g, h}.

Since the R-charge indices should be fully contracted, the $\eta_i$ orders of the terms in $\Ah$ should be even numbers. Otherwise, there will be at least one $\eta_i$ that cannot find a partner to contract with. At each even orders of $\eta_i^{(n)}$, the possible terms are shown, respectively,
\ie
\eta^{(0)}:\;& 1\\
\eta^{(2)}:\;& \HHa, \HHb, \HHab \\
\eta^{(4)}:\;& \HHa\HHb
\;,\fe
and the most general form of $\Ah$ is
\ie
\label{eqn: general form of h}
\Ah=c_0+c_1\HHa+c_2\HHab+c_3\HHb+c_4\HHa\HHb
\;.\fe
Demanding $\tilde{\mathcal{D}}_{\dot{\alpha} A}\Ah=0$, we get
\ie
\DbRaise{A}\Ah^{(0)}=\DbLower{A}\Ah^{(2)}\;;\;\DbRaise{A}\Ah^{(2)}=\DbLower{A}\Ah^{(4)}
\;,\fe
and resorting to \eqref{eqn: raising eta order} and \eqref{eqn: lowering eta order}, we arrive at the relations between $c_i$'s:
\ie
\label{eqn: relations of c_i}
c_2=-\frac{2\SinT}{\CosT}c_0,\;c_4=\frac{\SinTsqr}{\CosTsqr},\;c_1=0
\;.\fe
This implies there are two solutions, since according to \eqref{eqn: relations of c_i}, the general form \eqref{eqn: general form of h} is decoupled into two linearly independent terms (the result should not be surprising, since there are also two solutions in the massless case),
\ie
\label{eqn: two solutions of h}
\Ah^{(1)}&=1-{2\SinT\over \CosT}\HHab+{\SinTsqr\over \CosTsqr}\HHa \HHb \\
\Ah^{(2)}&=\HHb
\;.\fe
Substitute \eqref{eqn: two solutions of h} into $\mathcal{D}_\alpha^A \Ah-\epsilon^{AC}\lambda_{P\alpha}\Ag_C=0$, we obtain the $\Ag_A$ part of each solution
\ie
\label{eqn: two solutions of g_A}
\Ag^{(1)}_{A}&={1\over \CosT}\Ha_A+\frac{\SinT}{\CosTsqr}\HHa\Hb_A\\
\Ag^{(2)}_{A}&=\CosT\Ha_A\HHb+\SinT\Hb_A
\;.\fe
Last, substitute \eqref{eqn: two solutions of g_A} into $\mathcal{D}_\alpha^A g_C-\delta^A_C\lambda_{P\alpha}\Af=0$, we obtain the $\Af$ part of each solution
\ie
\Af^{(1)}&={1\over \CosTsqr}\HHa\\
\Af^{(2)}&=\SinTsqr+2\CosT\SinT\HHab+\CosTsqr\HHa\HHb
\;.\fe
Our final solutions are
\ie
\label{eqn: final solution (a)}
S^{(a)}(\tZ)=C^{(a)}\cdot \delta^1\delta^2\Big\{\Af^{(a)}+\eta_P^A\Ag^{(a)}_A+\frac{1}{2}\eta_P^A\eta_{PA}\Ah^{(a)}\Big\}
\;,\fe
where $a=1,2$, labeling the two solutions, and $C^{(1)}=\CosTsqr$, $C^{(2)}=-x$, so as to make both solutions LG neutral \footnote{One may wonder why adding $x$-factors is the only way to change the helicity of vacuum state of the massless multiplet. There are after all several quantities that carries helicity, including $\lambda_{P\alpha}$ and $\tilde\lambda_{P\dot\alpha}$, and of course the $x$-factor. However, both $\lambda_{P\alpha}$ and $\tilde\lambda_{P\dot\alpha}$ carry an extra $SL(2,C)$ index, i.e., $\alpha$ or $\dot\alpha$, and should be contracted with $\lambda_{i\alpha}^I$ or $\tilde\lambda_{i\dot\alpha}^I$ ($i=1,2$). But all $\lambda_{i\alpha}^I$ or $\tilde\lambda_{i\dot\alpha}^I$ dependence are encoded in $S^{(1)}$ and $S^{(2)}$, and their $SL(2,C)$ indices are already contracted. As a result, $x$-factors is the only choice remains.} (and also make $S^{(1)}$ and $S^{(2)}$ look more "symmetrical"). The solutions depend on $\CosT$ and $\SinT$, which are functions of $Z$.

Up to this point, we have two solutions to $\AAA$. However, they are written in terms of $\Ha^A$ and $\Hb_A$, which are not manifestly Lorentz covariant quantities themselves, see \eqref{eqn: expressing via reference spinors}. Only when they are combined with the delta functions $\delta^A$ can they be represented in a Lorentz covariant form. Let's first define\footnote{Note that $Q_{\alpha}^A$ and $\tQ_{\dot{\alpha}A}$ (normal character Q) stand for generators, while $\Math{Q}_{\alpha}^A$ and $\tilde{\Math{Q}}_{\dot{\alpha}}^A$ (curly character Q) stand for the little group covariant supersymmetric components of the solutions.}
\ie
&\Math{Q}_{\alpha}^A\equiv\lambda_{1\alpha}^I\eta_{1I}^A+\lambda_{2\alpha}^I\eta_{2I}^A\;\;;\;\;
\tilde{\Math{Q}}_{\dot{\alpha}}^A\equiv\tilde{\lambda}_{1\dot{\alpha}}^I\eta_{1I}^A-\tilde{\lambda}_{2\dot{\alpha}}^I\eta_{2I}^A\\
&\Math{Q}_{P\alpha}^A\equiv\lambda_{P\alpha}\eta_{P}^A\;\;;\;\;
\tilde{\Math{Q}}_{P\dot{\alpha}}^A\equiv\tilde{\lambda}_{P\dot{\alpha}}\eta_{P}^A\\
&\eta_i^A\cdot\eta_i^B\equiv-\frac{1}{2}\epsilon^{IJ}\eta_{iI}^A\eta_{iJ}^B\;\;;\;\;\eta_P\cdot\eta_P\equiv\frac{1}{2}\eta_P^A\eta_{PA}
\;.\fe
By combining $\Af^{(a)}$, $\Ag^{(a)}_A$, and $\Ah^{(a)}$ with the delta functions, we can finally write down the three point amplitudes in a Lorentz covariant and R-charge symmetric form. (For more details, see \ref{appendix: More on Ha and Hb}).\\
The first solution is
\ie
\label{S1: the first solution}
S^{(1)}(\tZ)=&\CosTsqr\cdot\delta^1\delta^2\Big\{ {1\over 2}\eta_P^A\eta_{PA} \Ah^{(1)} + \eta_P^A \Ag^{(1)}_{A} + \Af^{(1)} \Big\}\\
=&\frac{\CosTsqr}{4}\LRA{\Math{Q}^{A} \Math{Q}_P^{B}}\LRA{\Math{Q}_{A} \Math{Q}_{PB}}
+\frac{\CosT \SinT}{6x}\LRA{\Math{Q}^{A} \Math{Q}^{B}}\LRB{\tilde{\Math{Q}}_{A} \tilde{\Math{Q}}_{B}} \left( \eta_P\cdot\eta_{P} \right)
\\
&+\frac{\SinTsqr}{12x^2} \LRB{\tilde{\Math{Q}}^{A} \tilde{\Math{Q}}_{P}^{B}} \LRB{\tilde{\Math{Q}}_{A} \tilde{\Math{Q}}_{PB}} \left(\ETAsU{C}{D} \right)\left(\ETAsL{C}{D} \right)
\\
&+\frac{\CosT}{3}\LRA{\Math{Q}^{A} \Math{Q}^{B}}\LRA{\Math{Q}_{A} \Math{Q}_{PB}}
-\frac{2\SinT}{9x}\LRA{\Math{Q}^{A} \Math{Q}^{B}}  \LRB{\tilde{\Math{Q}}_{A} \tilde{\Math{Q}}_{P}^{C}}\left(\ETAsL{B}{C}\right)
\\
&+\frac{1}{12}\LRA{\Math{Q}^{A} \Math{Q}^{B}}\LRA{\Math{Q}_{A} \Math{Q}_{B}}
\;,\fe
and the second solution is
\ie
\label{S2: the second solution}
S^{(2)}(\tZ)=&-x\cdot\delta^1\delta^2\Big\{ {1\over 2}\eta_P^A\eta_{PA} \Ah^{(2)} + \eta_P^A \Ag^{(2)}_{A} + \Af^{(2)} \Big\}
\\
=&\frac{1}{12x}\LRB{\tilde{\Math{Q}}^{A} \tilde{\Math{Q}}^{B} }\LRB{\tilde{\Math{Q}}_A\tilde{\Math{Q}}_B}\left( \eta_P\cdot\eta_{P} \right)
-\frac{2\CosT}{9x}\LRB{\tilde{\Math{Q}}^A \tilde{\Math{Q}}^{B}}   \LRB{\tilde{\Math{Q}}_A\tilde{\Math{Q}}_P^{C}} \left( \ETAsL{B}{C} \right)
\\
&+\frac{\SinT}{3} \LRB{\tilde{\Math{Q}}^{A}\tilde{\Math{Q}}^{B}}\LRA{\Math{Q}_A\Math{Q}_{PB}}
+\frac{\SinTsqr x}{2}\LRA{\Math{Q}^AP}\LRA{\Math{Q}_AP}
+\frac{\CosT\SinT}{6}\LRA{\Math{Q}^{A}\Math{Q}^{B}}\LRB{\tilde{\Math{Q}}_A\tilde{\Math{Q}}_B}
\\
&+\frac{\CosTsqr}{6x} \LRB{\tilde{\Math{Q}}^AP} \LRB{\tilde{\Math{Q}}_AP} \left( \ETAsU{B}{C} \right) \left( \ETAsL{B}{C} \right)
\;.\fe
Note that both of the solutions are inhomogeneous in Grassmann degree. The first term in $S^{(1)}(\tZ)$ has Grassmann degree 4, followed by terms with Grassmann degree 6,8,4,6,4. Similar thing happens to $S^{(2)}(\tZ)$.

The two solutions are related. If we do Grassmann Fourier transform (defined in~\cite{Herderschee:2019ofc}) to either of the solutions, and change all angle brackets to square brackets (and vice versa), we get the other solution (up to an overall constant). This is in fact a direct consequence of the form of the generators \eqref{eqn: QQ+ full}. If we change all $\eta$'s into $\frac{\partial}{\partial\eta}$'s (and vice versa) in $Q_{\alpha}$, and change all $\lambda$'s into $\tilde{\lambda}$'s, we get exactly $\tQ_{\dot{\alpha}}$. We can do the same transformation to $\tQ_{\dot{\alpha}}$ to get $Q_{\alpha}$. Since the $\AAA$ is invariant under transformation generated by both $Q_{\alpha}$ and $\tQ_{\dot{\alpha}}$, if the generators are symmetric under certain transformation, the solutions would inherit this property, and come in pairs consequently. The symmetry of the generators is the origin of why the solutions must come in pairs.

Observe that the solution \eqref{S1: the first solution} can be factorized into a product of two components
\ie
\label{eqn: S1 factorized pieces}
S^{(1)}(\tZ)=&\left( \CosT\eta_P^1\delta^1-\delta^1\Ha^1-\CosT\SinT\delta^1\Hb_2\eta_P^1\eta_P^2+\SinT\delta^1\Ha^1\Hb_2\eta_P^2 \right)\\
\cdot&\left( \CosT\eta_P^2\delta^2-\delta^2\Ha^2+\CosT\SinT\delta^2\Hb_1\eta_P^1\eta_P^2+\SinT\delta^2\Ha^2\Hb_1\eta_P^1 \right)
\;.\fe
The solution \eqref{S2: the second solution} can be factorized in a similar expression
\ie
\label{eqn: S2 factorized pieces}
S^{(2)}(\tZ)=-x\cdot
&\left( \CosT\delta^1\Ha^1\Hb_1+\delta^1\eta_P^1\Hb_1+\SinT\delta^1 \right)\\ \cdot
&\left( \CosT\delta^2\Ha^2\Hb_2+\delta^2\eta_P^2\Hb_2+\SinT\delta^2 \right)
\;.\fe
This observation will later play a crucial role in the discussion of BPS limit.

To sum up, $\AAA$ is the linear combination of the two solutions
\begin{tcolorbox}
\ie
\label{eqn: linear combination A}
\AAA(Z,m)=\alpha_1\cdot S^{(1)}(\tZ)+\alpha_2\cdot S^{(2)}(\tZ)
\;.\fe
\end{tcolorbox}
Note that all $\eta$ dependencies are encoded in $S^{(1)}$ and $S^{(2)}$. There are two unfixed coefficients $\alpha_1$ and $\alpha_2$, which can be fixed by requiring parity symmetry. We will go back to this later.

\subsection{Spinning multiplets}
\label{subsection: spinning multiplets}
A 3-pt amplitude of two equal massive particles with mass $m$ and a massless particle with helicity $h$ can be written in spinor helicity basis, see~\cite{Arkani-Hamed:2017jhn},
\ie
M^{(I_1\cdots I_{2s_1})(J_1\cdots J_{2s_2})h}=
\left(\lambda_1\right)^{I_1}_{\alpha_1}\cdots\left(\lambda_1\right)^{I_{s_1}}_{\alpha_{s_1}}
\left(\lambda_2\right)^{I_2}_{\beta_2}\cdots\left(\lambda_2\right)^{I_{s_2}}_{\beta_{s_2}}
M^{(\alpha_1\cdots \alpha_{2s_1})(\beta_1\cdots \beta_{2s_2})h}
\;,\fe
where
\ie
M^{(\alpha_1\cdots \alpha_{2s_1})(\beta_1\cdots \beta_{2s_2})h}
=\sum_{i=|s_1-s_2|}^{(s_1+s_2)} g_i x^h \left[ \lambda_P^i \left( \frac{p_1\tilde{\lambda}_P}{m} \right)^i \epsilon^{s_1+s_2-i} \right]^{(\alpha_1\alpha_2\cdots\alpha_{2s_1}),(\beta_1\beta_2\cdots\beta_{2s_2})}
\;.\fe
The amplitude for spinning multiplets can always be factorized into a product of a bosonic part $\Msasb{s_1}{s_2}{h}$ and $\AAA$,
\begin{tcolorbox}
\ie
\label{eqn: spinning multiplet the most general form}
&\sMsasb{s_1}{s_2}{h}\\
&\;=\Msasb{s_1}{s_2}{h}\cdot\left[\alpha_1\cdot S^{(1)}(\tZ)+\alpha_2\cdot S^{(2)}(\tZ)\right]
\;.\fe
\end{tcolorbox}
\noindent
In other words, if a spinor helicity variable that carries LG indices of the massive multiplets (e.g., $I_1\cdots I_{2s_1}$) have their $SL(2,C)$ indices (e.g., $\alpha$ in $\lambda_\alpha^I$) contract with a spinor helicity variable whose LG indices contract with a Grassmann variable $\eta_{iI}^A$ (we say the amplitude is "polluted"), then we can always use Schouten identity to reorder the $SL(2,C)$ indices to make the amplitude "unpolluted". More explicitly, we can always write the amplitude without terms such as $\LRA{\bo^{I_1}\bo^I}\eta_{1I}^A$. The reason is if such terms do exist, then according to the discussions in \ref{appendix: amplitude = f(Ha,Hb)}, they must come in combinations $\lambda_{1\alpha}^I\eta_{1I}^A+\lambda_{2\alpha}^I\eta_{2I}^A$ or $\tilde\lambda_{1\dot\alpha}^I\eta_{1I}^A-\tilde\lambda_{2\dot\alpha}^I\eta_{2I}^A$, and can be re-written in terms of $\Ha$ or $\Hb$. For example, if $\LRA{\bo^{I_1}\bo^I}\eta_{1I}^A$ do exist, then it must appear in the form of $\sum_{i}\LRA{\bo^{I_1}\bi^I}\eta_{iI}^A$; but according to \eqref{eqn: expressing via reference spinors}, $\sum_{i}\LRA{\bo^{I_1}\bi^I}\eta_{iI}^A = \LRA{\bo^{I_1} P}\Ha^A$, and thus make the term "unpolluted".

\paragraph{Parity}
$ $\newline
Parity symmetry relates each amplitude and its conjugate amplitude, demanding they have the same couplings, and that constraints the coefficients $\alpha_1$ and $\alpha_2$ in \eqref{eqn: spinning multiplet the most general form}. In $\Math{N}=2$ supersymmetry, the massless multiplets with helicity $+h$ and $-h+1$ are
\ie
&{\Pmulti{+h}}=\field{h}+\eta^A\field{h-\tfrac{1}{2}}_A+\eta^A\eta_A\field{h-1}\\
&{\Pmulti{-h+1}}=\field{-h+1}+\eta^A\field{-h+\tfrac{1}{2}}_A+\eta^A\eta_A\field{-h}
\;.\fe
We can see the components are related by parity transformation, and parity invariance requires the coefficients of the superamplitudes with $(+h)$-helicity and $(-h+1)$-helicity multiplet should be related. Let's first write down the form of the superamplitudes
\ie
&\sMsasb{s}{s}{h}=\Msasb{s}{s}{h}\left[\alpha_1\cdot S^{(1)}+\alpha_2\cdot S^{(2)}\right]\\
&\sMsasb{s}{s}{-h+1}=\Msasb{s}{s}{-h+1}\left[\beta_1\cdot S^{(1)}+\beta_2\cdot S^{(2)}\right]
\;,\fe
where
\ie
\Msasb{s}{s}{+2}=&g_0x^2\prod_{i,j=1}^{2s}\LRA{\bo^{(I_i)}\bt^{(J_j)}}+g_1\frac{x}{m}\prod_{i,j=1}^{2s-1}\LRA{\bo^{(I_i}\bt^{(J_j}}\LRA{\bo^{2s)}P}\LRA{\bt^{{2s)}}P}\\
+&g_2\frac{1}{m^2}\prod_{i,j=1}^{2s-2}\LRA{\bo^{(I_i}\bt^{(J_j}}\prod_{i,j=2s-1}^{2s}\LRA{\bo^{I_i)}P}\LRA{\bt^{ J_j)}P}+\cdots
\\
\barMsasb{s}{s}{0}=&\bar g_0\prod_{i,j=1}^{2s}\LRB{\bo^{(I_i)}\bt^{(J_j)}}+\bar g_1\frac{1}{xm}\prod_{i,j=1}^{2s-1}\LRB{\bo^{(I_i}\bt^{(J_j}}\LRB{\bo^{2s)}P}\LRB{\bt^{{2s)}}P}\\
+&\bar g_2\frac{1}{x^2m^2}\prod_{i,j=1}^{2s-2}\LRB{\bo^{(I_i}\bt^{(J_j}}\prod_{i,j=2s-1}^{2s}\LRB{\bo^{I_i)}P}\LRB{\bt^{ J_j)}P}+\cdots
\;.\fe
Parity invariance requires $(\beta_1,\beta_2)=(\alpha_2,\alpha_1)$ and $g_i=\bar g_i$.


\section{Limits of the amplitude}
\label{section: Z=0 and BPS and massless}

In this section, the $Z=0$ limit and the BPS limit of $\AAA$ are examined. The former is a product of two $\Math{N}=1$ SUSY invariant delta functions, and the later is a product of two components that are related to each other by Grassmann Fourier transformation.

\subsection{The $Z=0$ limit}
When $Z_1=Z_2=0$, $\tZ\rightarrow0$, and the solution takes the form (see \eqref{eqn: S1 factorized pieces})
\ie
\label{Z=0 limit MHV amplitude}
S^{(1)}=\prod_{A=1}^2\Bigg[\LRA{\bo^I\bt^J}\eta_{1I}^A\eta_{2J}^A+\sum_{i=1}^2\LRA{\bi^IP}\eta_{iI}^A\eta_P^A+\sum_{i=1}^2m\etas{i}{A}\Bigg]
\;,\fe
where we defined
\ie
\etas{i}{A}\equiv-{1\over2}\epsilon^{IJ}\eta_{iI}^A\eta_{iJ}^A \;\;\text{(no sum over $i$ and $A$)}
\;.\fe
This is exactly the square of $\Math{N}=1$ SUSY invariant delta function defined in~\cite{Herderschee:2019ofc}, which is not surprising, since we can see from \eqref{eqn: QQ+ full} that $Q_\alpha^A$ is purely multiplicative in $\eta$'s, while $\tQ_{\dot{\alpha}A}$ is purely differential in $\eta$'s, and there is no mixing terms between two R-charge indices. Therefore, the solution being multiplicative of two solutions corresponding to two R-charges is an expected result.
The other solution is
\ie
\label{Z=0 limit MHV_bar amplitude}
S^{(2)}=\prod_{A=1}^2\Bigg[&\LRB{\bo^I\bt^J}\eta_{1I}^A\eta_{2J}^A\eta_P^A+\LRB{\bo^IP}\eta_{1I}^A\etas{2}{A}+\LRB{\bt^IP}\eta_{2I}^A\etas{1}{A}\\
&+m\eta_P^A\etas{1}{A}+m\eta_P^A\etas{2}{A}\Bigg]
\;,\fe
which is related to \eqref{Z=0 limit MHV amplitude} by Grassmann Fourier transformation.

In the high energy limit, we can take $\eta_{i1}^A\rightarrow\eta_i^A$, $\eta_{i2}^A\rightarrow\tilde\eta_i^A$, and the above amplitudes have the limits\footnote{In the massless limit, the number of supercharges is reduced in half, resulting residual Grassmann variables $\tilde\eta_i$. Those residual Grassmann variables divide the massive superfield into different massless representations \cite{Herderschee:2019ofc}.} (see \cite{Arkani-Hamed:2017jhn,Herderschee:2019ofc})
\ie
&S^{(1)}\rightarrow\prod_{A=1}^2\Bigg[\LRA{12}\eta_{1}^A\eta_{2}^A+\sum_{i=1}^2\LRA{1P}\eta_{iI}^A\eta_P^A\Bigg]
\\
&S^{(2)}\rightarrow\prod_{A=1}^2\Bigg[\LRB{12}\eta_P^A+\LRB{P1}\eta_{2}^A+\LRB{2P}\eta_{1}^A\Bigg]\tilde\eta_1^A\tilde\eta_2^A.
\fe
We can see they are equal to the MHV and anti-MHV amplitudes in \cite{Elvang:2015rqa}.

\subsection{The BPS limit}

When the massive particles saturate the BPS limit, i.e., $Z_1=Z_2=2m$, $\tZ\rightarrow\pi/4$, and \eqref{eqn: QQ+ full} becomes
\ie
\label{eqn: BPS limit generators}
&Q_\alpha^A=\sum_{i=1}^2\lambda_{i\alpha}^I\left({1\over\sqrt{2}}\eta_{iI}^A+{1\over\sqrt{2}}\sigma_i\frac{\partial}{\partial\eta_{iA}^I}\right)+\lambda_{P\alpha}\eta_P^A
\\
&\tQ_{\dot{\alpha}A}=\sum_{i=1}^2\tilde{\lambda}_{i\dot{\alpha}}^I\left({1\over\sqrt{2}}\frac{\partial}{\partial\eta_{i}^{IA}}-{1\over\sqrt{2}}\sigma_i\eta_{iIA}\right)+\tilde{\lambda}_{P\dot{\alpha}}\frac{\partial}{\partial\eta_P^A}
\;.\fe
Although the anti-commutation relations \eqref{eqn: the anti-commutation relations- massive} still hold, \eqref{eqn: linear combination A} is no longer valid. To see this, note that $\AAA$ is proportional to the delta function \eqref{eqn: the delta function} if and only if $\CosTsqr-\SinTsqr=0$ (See subsection \ref{subsection:  the delta function}), and the BPS limit doesn't meet this criteria. As a result, our solutions \eqref{eqn: linear combination A} do not apply to the BPS limit. However, in the BPS limit, the $\eta_i$ degree of freedom is half the number of that in the non-BPS limit. Therefore, to obtain the amplitude of the BPS limit, we shouldn't have started with \eqref{eqn: QQ+ full}, but rather one of the following set (either of them is sufficed to serve as a set of generators in the BPS limit),
\ie
\label{eqn: BPS limit two sets of generators}
\begin{cases}
Q_\alpha^1=\sum_{i=1}^2\lambda_{i\alpha}^I\eta_{iI}+\lambda_{P\alpha}^I\eta_{P}^1\\
Q_\alpha^2=\sum_{i=1}^2\sigma_i\lambda_{i\alpha}^I\frac{\partial}{\partial\eta_{i}^{I}}+\lambda_{P\alpha}^I\eta_{P}^2\\
\tQ_{\dot{\alpha}1}=\sum_{i=1}^2\tilde{\lambda}_{i\alpha}^I\frac{\partial}{\partial\eta_{i}^{I}}+\tilde{\lambda}_{P\dot{\alpha}}^I\frac{\partial}{\partial\eta_{P}^1}\\
\tQ_{\dot{\alpha}2}=-\sum_{i=1}^2\sigma_i\tilde{\lambda}_{i\alpha}^I\eta_{iI}+\tilde{\lambda}_{P\dot{\alpha}}^I\frac{\partial}{\partial\eta_{P}^2}
\end{cases}
\begin{cases}
Q_\alpha^1=-\sum_{i=1}^2\sigma_i\lambda_{i\alpha}^I\frac{\partial}{\partial\eta_{i}^{I}}+\lambda_{P\alpha}^I\eta_{P}^1\\
Q_\alpha^2=\sum_{i=1}^2\lambda_{i\alpha}^I\eta_{iI}+\lambda_{P\alpha}^I\eta_{P}^2\\
\tQ_{\dot{\alpha}1}=\sum_{i=1}^2\sigma_i\tilde{\lambda}_{i\alpha}^I\eta_{iI}+\tilde{\lambda}_{P\dot{\alpha}}^I\frac{\partial}{\partial\eta_{P}^1}\\
\tQ_{\dot{\alpha}2}=\sum_{i=1}^2\tilde{\lambda}_{i\alpha}^I\frac{\partial}{\partial\eta_{i}^{I}}+\tilde{\lambda}_{P\dot{\alpha}}^I\frac{\partial}{\partial\eta_{P}^2}
\end{cases}
\;.\fe

Although \eqref{S1: the first solution} and \eqref{S2: the second solution} are not the solutions for BPS limit, interesting things happen if we take $\tZ=\pi/4$. In this limit, the two factorized components in \eqref{eqn: S1 factorized pieces} are related. If we do Grassmann Fourier transformation to one of the components, then complex conjugate spinor helicity variables ($\lambda\leftrightarrow\tilde\lambda$), and transform $\eta_P$ properly, it will become the other component,
\ie
&\left( \frac{1}{\sqrt{2}}\eta_P^2\delta^2-\delta^2\Ha^2+\frac{1}{2}\delta^2\Hb_1\eta_P^1\eta_P^2+\frac{1}{\sqrt{2}}\delta^2\Ha^2\Hb_1\eta_P^1 \right) \Bigg|_{\substack{\lambda\leftrightarrow\tilde\lambda;\;\eta\rightarrow\bar\eta\\ \bar\eta_P \rightarrow \sqrt{2}\bar\eta_P}}\\
&=\text{FT}\left[\left( \frac{1}{\sqrt{2}}\eta_P^1\delta^1-\delta^1\Ha^1-\frac{1}{2}\delta^1\Hb_2\eta_P^1\eta_P^2+\frac{1}{\sqrt{2}}\delta^1\Ha^1\Hb_2\eta_P^2 \right)\Bigg|_{\bar\eta_P \rightarrow \sqrt{2}\bar\eta_P}\right]
\;,\fe
where
\ie
\text{FT}\left[f(\eta)\right]\equiv\frac{1}{8} \left[ \prod_{i=1}^2\int d\eta_{iI}^1d\eta_{i}^{I1} e^{\eta_{iI}^2\eta_{i}^{I1}}\right]\left[ \int d\eta_{P}^Ad\eta_{PA} e^{\bar\eta_{P}^A\eta_{PA}}\right]f(\eta)
\;.\fe
By comparing with \eqref{eqn: fourier transform-massive}, we see that $\eta_{iI}^2$ act as the barred Grassmann number of $\eta_{iI}^1$, and serve as another $\eta$-basis. In other words, \eqref{S1: the first solution} has unnecessary Grassmann variables in the BPS limit. We can use either of the factorized components in \eqref{eqn: S1 factorized pieces} to describe a BPS amplitude.

Let's see how the components in \eqref{eqn: S1 factorized pieces} are related to BPS amplitudes. If we rename the $\eta_{iI}$ in the LHS of \eqref{eqn: BPS limit two sets of generators} as $\eta_{iI}^1$, and the $\eta_{iI}$ in the RHS as $\eta_{iI}^2$, then the SUSY invariant amplitudes are (the first of them is the solution of the LHS set, while the second is that of the RHS set)~\cite{Caron-Huot:2018ape},
\ie
\label{eqn: BPS solutions}
\Math{A}_{BPS}^1=\eta_P^1\delta^1-\delta^1\Ha^1-\delta^1\Hb_2\eta_P^1\eta_P^2+\delta^1\Ha^1\Hb_2\eta_P^2\;\;\\
\Math{A}_{BPS}^2=\eta_P^2\delta^2-\delta^2\Ha^2+\delta^2\Hb_1\eta_P^1\eta_P^2+\delta^2\Ha^2\Hb_1\eta_P^1
\;,\fe 
where the definition of $\delta^A$, $\Ha^A$, $\Hb_A$ are same as the previous ones\footnote{Note that in~\cite{Caron-Huot:2018ape}, the amplitude reads \ie \left( -\sum_{i,j} \LRA{\bi^I\bj^J}\eta_{iI}^1\eta_{jJ}^1-2\sum_{i}\LRA{\bi^IP}\eta_{iI}^1\eta_P \right) \cdot \exp\left( \frac{1}{\LRB{\zeta P}} \sum_{i} \sigma_i\LRB{\zeta \bi^I}\eta_{iI}^1\eta_P^2 \right) \;,\fe which, according to \eqref{eqn: expressing via reference spinors}, equals to \ie  \Math{A}_{BPS}^1=(\delta^1\Ha^1-\eta_P^1\delta^1)\cdot\exp\left( \Hb_2\eta_p^2 \right)=\eta_P^1\delta^1-\delta^1\Ha^1-\delta^1\Hb_2\eta_P^1\eta_P^2+\delta^1\Ha^1\Hb_2\eta_P^2 \;.\fe}. Note that the index $A$ in the previous non-BPS calculations stands for R-charge, while in \eqref{eqn: BPS solutions}, $A$ is a label, labeling the two solutions. More concretely, $\eta_{iI}^1$ and $\eta_{iI}^2$ are more like a conjugate pair in \eqref{eqn: BPS solutions}. Despite the meaning of $A$ in BPS limit departs radically from non-BPS case, the solutions in \eqref{eqn: BPS solutions} can be related to non-BPS solution. The product of the two solutions in \eqref{eqn: BPS solutions} is
\ie
\Math{A}_{BPS}^1\times\Math{A}_{BPS}^2=&
\eta_P^1\delta^1\eta_P^2\delta^2+2\delta^1\delta^2\eta_P^1\eta_P^2\HHab-\delta^1\delta^2\eta_P^1\eta_P^2\HHa\HHb\\
&+\delta^1\delta^2\eta_P^A\Ha_A+\delta^1\delta^2\HHa\eta_P^A\Hb_A+\delta^1\delta^2\HHa
\;.\fe
If we replace all $\eta_P^A$ replaced by $\frac{1}{\sqrt{2}}\eta_P^A$, it will be equal to $S^{(1)}(\tZ=\frac{\pi}{4})$,
\ie
\label{eqn: BPS and non-BPS}
\Math{A}_{BPS}^1\times\Math{A}_{BPS}^2 \Big|_{\eta_P^A \rightarrow \frac{1}{\sqrt{2}}\eta_P^A}=S^{(1)}(\tZ)
\;.\fe

One may wonder why there should be a $\eta_P\rightarrow\frac{1}{\sqrt{2}}\eta_P$ in \eqref{eqn: BPS and non-BPS}, and we can trace its origin from the generators \eqref{eqn: QQ+ full}. In the limit $\tZ=\frac{\pi}{4}$, if we sum the two sets of generators $Q_\alpha^A$ in \eqref{eqn: BPS limit two sets of generators}, and modify $\eta_P$ by $\frac{1}{\sqrt{2}}\eta_P$, it will be proportional to the generators $Q_\alpha^A$ in \eqref{eqn: QQ+ full}. This explains why there should be a $\frac{1}{\sqrt{2}}$ factor. One may argue that once we make $\eta_P\rightarrow\frac{1}{\sqrt{2}}\eta_P$, the differential part in $\tQ_{\dot{\alpha}}$ will acquire a $\sqrt{2}$ factor, so that the sums of $Q_{\dot{\alpha}}^A$ in \eqref{eqn: BPS limit two sets of generators} are not proportional to $Q_{\dot{\alpha}}^A$ in \eqref{eqn: QQ+ full}. But since non-BPS solution is the product of BPS solutions, when applying a differential operator, we would need to impose the product law of differentiation, which induces an extra factor of $2$, and $\frac{1}{\sqrt{2}}\times2=\sqrt{2}$ explains the $\sqrt{2}$ factor.


\section{$\Math{N}=4$ Super-Maxwell and Supergravity}
\label{section: N=4 super-Maxwell and SUGRA}

We have been working on $\Math{N}=2$ SUSY amplitudes in the previous sections, and now we consider $\Math{N}=4$ SUSY, where the central charge matrix can be put in the standard block-diagonal form, 
\ie
\label{eqn: N=4 central charge}
Z_{AB}=
\begin{bmatrix}
0&-Z_{12}&0&0\\
Z_{12}&0&0&0\\
0&0&0&-Z_{34}\\
0&0&Z_{34}&0
\end{bmatrix}
\;.\fe
In central charge free $\Math{N}=4$ SUSY, the R-symmetry group is $SU(4)$. However, central charge extension breaks the full $SU(4)$ group into $SU(2)\otimes SU(2)$ subgroup.

Since the $\Math{N}=1,2$ part and the $\Math{N}=3,4$ part do not mix with each other, the supersymmetric part of the amplitude is the square of $\AAA$, with $\tZ$'s properly modified (c.f. \eqref{eqn: spinning multiplet the most general form}),
\ie
\label{eqn: spinning multiplet the most general form, N=4}
&\sMZaZbsasb{s_1}{s_2}{h}=\\
&\;\;\Msasb{s_1}{s_2}{h}\cdot
\left[\alpha_1 S_{12}^{(1)}(\tZa)+\alpha_2 S_{12}^{(2)}(\tZa)\right]
\left[\alpha_3 S_{34}^{(1)}(\tZb)+\alpha_4 S_{34}^{(2)}(\tZb)\right]
\;,\fe
where
\ie
2\CosTa\SinTa=\frac{Z_{12}}{2m}\;\;;\;\;2\CosTb\SinTb=\frac{Z_{34}}{2m}
\;.\fe
The subscripts that $S_{AB}^{(1)}$ and $S_{AB}^{(2)}$ carries, i.e., $12$ or $34$, indicate which projected $SU(2)$ group they are describing.

\paragraph{Parity}
$ $\newline
As we discussed in subsection \ref{subsection: spinning multiplets}, the coefficients $\alpha_n$ of different superamplitudes are related in a parity invariant theory. In $\Math{N}=4$ SUSY, similar things happen. First note that the massless multiplet with helicity-$(+h)$ and helicity-$(-h+2)$ are related by parity
\ie
\label{eqn: spectrum}
{\Pmulti{+h}}
=&\field{+h}
+\eta^{A^\prime}\field{h-\tfrac{1}{2}}_{A^\prime}
+\eta^{A^{\prime\prime}}\field{h-\tfrac{1}{2}}_{A^{\prime\prime}}\\
+&\eta^{A^\prime}\eta_{A^{\prime}} \field{h-1}_{12}
+\eta^{A^{\prime\prime}}\eta_{A^{\prime\prime}} \field{h-1}_{34}
+\eta^{A^\prime}\eta^{A^{\prime\prime}} \field{h-1}_{A^\prime A^{\prime\prime}}\\
+&\eta^{A^{\prime\prime}}\eta_{A^{\prime\prime}}\eta^{A^{\prime}}\field{h-\tfrac{3}{2}}_{A^{\prime}}
+\eta^{A^\prime}\eta_{A^\prime}\eta^{A^{\prime\prime}}\field{h-\tfrac{3}{2}}_{A^{\prime\prime}}
+\eta^{A^{\prime}}\eta_{A^{\prime}}\eta^{A^{\prime\prime}}\eta_{A^{\prime\prime}}\field{h-2}
\\
{\Pmulti{-h+2}}
=&\field{-h+2}
+\eta^{A^\prime}\field{-h+\tfrac{3}{2}}_{A^\prime}
+\eta^{A^{\prime\prime}}\field{-h+\tfrac{3}{2}}_{A^{\prime\prime}}\\
+&\eta^{A^\prime}\eta_{A^{\prime}} \field{-h+1}_{12}
+\eta^{A^{\prime\prime}}\eta_{A^{\prime\prime}} \field{-h+1}_{34}
+\eta^{A^\prime}\eta^{A^{\prime\prime}} \field{-h+1}_{A^\prime A^{\prime\prime}}\\
+&\eta^{A^{\prime\prime}}\eta_{A^{\prime\prime}}\eta^{A^{\prime}}\field{-h+\tfrac{1}{2}}_{A^{\prime}}
+\eta^{A^\prime}\eta_{A^\prime}\eta^{A^{\prime\prime}}\field{-h+\tfrac{1}{2}}_{A^{\prime\prime}}
+\eta^{A^{\prime}}\eta_{A^{\prime}}\eta^{A^{\prime\prime}}\eta_{A^{\prime\prime}}\field{-h}
\;,\fe
where $A^\prime=1,2$, $A^{\prime\prime}=3,4$. Therefore, the amplitude with $(+h)$-helicity multiplet and $(-h+2)$-helicity would have their coefficients $\alpha_n$ related. More concretely, if
\ie
\label{eqn: N=4, massive scalar multiplets, different helicities}
&\sMZaZbsasb{s}{s}{h}\\
&\;\;\;\;\;=\Msasb{s}{s}{h}\left[\alpha_1\cdot S_{12}^{(1)}+\alpha_2\cdot S_{12}^{(2)}\right]\left[\alpha_3\cdot S_{34}^{(1)}+\alpha_4\cdot S_{34}^{(2)}\right]\\
&\sMZaZbsasb{s}{s}{-h+2}\\
&\;\;\;\;\;=\barMsasb{s}{s}{-h+2}\left[\beta_1\cdot S_{12}^{(1)}+\beta_2\cdot S_{12}^{(2)}\right]\left[\beta_3\cdot S_{34}^{(1)}+\beta_4\cdot S_{34}^{(2)}\right]
\;,\fe
then parity invariance requires $(\beta_1,\beta_2,\beta_3,\beta_4)=(\alpha_2,\alpha_1,\alpha_4,\alpha_3)$.

If the superamplitude is self-conjugate, i.e., the massless multiplet is a spin-$1$ multiplet, parity invariance will require $\alpha_1=\alpha_2$, $\alpha_3=\alpha_4$, and therefore
\ie
\label{eqn: self-conjugate alpha_n}
&\sMZaZbsasb{s_1}{s_2}{+1}\\
&\;\;\;=\Msasb{s_1}{s_2}{+1}\cdot
\left[S_{12}^{(1)}(\tZa)+S_{12}^{(2)}(\tZa)\right]
\left[S_{34}^{(1)}(\tZb)+S_{34}^{(2)}(\tZb)\right]
\;.\fe

\subsection{$\Math{N}=4$ Super-Maxwell}
\label{subsection: N=4 Super-Maxwell}
If the massless multiplet carries helicity $+1$, then the helicities of the component fields span from $+1$ to $-1$, which describes Super-Maxwell theory. On the other hand, if the vacuum states of the massive particles are scalars, then the component amplitudes will have massive particles' spins up to spin $2$. Demanding that parity symmetry should be preserved, see \eqref{eqn: self-conjugate alpha_n}, we have
\ie
\label{Super-Maxwell amplitude}
\Math{M}(\bold{\Phi},\bold{\Phi},\Pmulti{+1})= \frac{e}{m^3} \cdot x \left[  S^{(1)}_{12}+ S^{(2)}_{12} \right]\left[ S^{(1)}_{34}+ S^{(2)}_{34} \right]
\;,\fe
where $e$ is the electrical charge and $S^{(1)}_{AB}\equiv S^{(1)}(\theta_{Z_{AB}})$. Note that $\theta_{Z_{AB}}$, which are functions of $Z_{AB}$, determines all the coefficients of the component amplitudes, including minimal and non-minimal couplings. Let's explicitly write down the form of the component amplitude ${M}(\bo^{s=2},\bt^{s=2},P^{+1})$ as an example, where both massive external states are spin 2, and the massless external state has helicity $+1$,
\ie
M(\bo^{s=2},\bt^{s=2},P^{+1})=
\frac{e}{m^3}\Bigg\{&\left[\left( 1-\frac{Z_{12}}{2m} \right)\left( 1-\frac{Z_{34}}{2m} \right)\right] x \LRA{12}^4\\
+&\left[ \frac{Z_{12}Z_{34}}{2m^2}-\frac{Z_{12}}{2m}-\frac{Z_{34}}{2m} \right] x\LRA{12}^3\frac{\MixLeft{1}{P}{2}}{m}\\
+&\left[ \frac{Z_{12}Z_{34}}{4m^2} \right] x\LRA{12}^2\left(\frac{\MixLeft{1}{P}{2}}{m}\right)^2\Bigg\}\\
\Rightarrow(g_0,g_1,g_2)=&\left( \left( 1-\frac{Z_{12}}{2m} \right)\left( 1-\frac{Z_{34}}{2m} \right) ,  \frac{Z_{12}Z_{34}}{2m^2}-\frac{Z_{12}}{2m}-\frac{Z_{34}}{2m} , \frac{Z_{12}Z_{34}}{4m^2} \right)
\;,\fe
with all the LG indices symmetrized. We can see that the component amplitude consists of non-minimal couplings. In addition, the non-minimal couplings of ${M}(\bo^{s=2},\bt^{s=2},P^{+1})$ vanish if and only if $Z_{AB}=0$.

\subsection{$\Math{N}=4$ Supergravity}
\label{subsection: N=4 SUGRA}
For $\Math{N}=4$ supergravity, since we have only 4 $\eta_P$'s, we can't obtain whole graviton spectrum (from $+2$ to $-2$) in one superamplitude. Rather, the spectrum is composed of two superamplitudes, one of which, $\Math{M}(\bold{\Phi},\bold{\Phi},\Pmulti{+2})$, has its massless spectrum range from $+2$ to $0$, and the other, $\Math{M}(\bold{\Phi},\bold{\Phi},\Pmulti{0})$, has it range from $0$ to $-2$. In general, $\Math{M}(\bold{\Phi},\bold{\Phi},\Pmulti{+2})$ and $\Math{M}(\bold{\Phi},\bold{\Phi},\phi)$ can be of the form
\ie
\label{eqn: N=4 SUGRA alpha beta}
&\Math{M}(\bold{\Phi},\bold{\Phi},\Pmulti{+2})=\frac{\kappa}{2m^2} x^2\left[ \alpha_1 S^{(1)}_{12}+ \alpha_2 S^{(2)}_{12} \right]\left[ \alpha_3 S^{(1)}_{34}+ \alpha_4 x S^{(2)}_{34} \right]\\
&\Math{M}(\bold{\Phi},\bold{\Phi},\Pmulti{0})= \frac{\kappa}{2m^2} \left[ \beta_1 S^{(1)}_{12}+ \beta_2 S^{(2)}_{12} \right]\left[ \beta_3 S^{(1)}_{34}+ \beta_4 x S^{(2)}_{34} \right]
\;.\fe
We require the following in order to preserve parity symmetry, see \eqref{eqn: N=4, massive scalar multiplets, different helicities},
\ie
\label{eqn: constraint 1 for N=4 SUGRA alpha beta}
(\alpha_1,\alpha_2,\alpha_3,\alpha_4)=(\beta_2,\beta_1,\beta_4,\beta_3)
\;.\fe

\paragraph{The $g_1$ problem of gravitational interaction}
$ $ \newline
The most general form of the 3-pt amplitude of two massive spinning particles with equal masses $m$ interacting with a massless particle, where the two massive particles are spin $S_1$ and $S_2$, and the massless particle has helicity $h$, is of the form (see ~\cite{Arkani-Hamed:2017jhn})
\ie
M({1}^{(I_1\cdots I_{2s})},{2}^{(J_1\cdots J_{2s})},P^{h})=&g_0x^{h}\prod_{i,j=1}^{2s}\LRA{\bo^{(I_i)}\bt^{(J_j)}}+g_1\frac{x^{h-1}}{m}\prod_{i,j=1}^{2s-1}\LRA{\bo^{(I_i}\bt^{(J_j}}\LRA{\bo^{2s)}P}\LRA{\bt^{{2s)}}P}\\
+&g_2\frac{x^{h-2}}{m^2}\prod_{i,j=1}^{2s-2}\LRA{\bo^{(I_i}\bt^{(J_j}}\prod_{i,j=2s-1}^{2s}\LRA{\bo^{I_i)}P}\LRA{\bt^{ J_j)}P}+\cdots
\;.\fe
It is shown in~\cite{Chung:2018kqs} that, if $h=2$ in the above equation, i.e., 3-pt amplitude includes a graviton, then $g_1$-term must vanish. Otherwise, we can't write down a Lagrangian with local operators that give raise to the amplitude.

Let's consider the component amplitude $M(\bo^{s=2},\bt^{s=2},P^{+2})$ (one of the component amplitudes of $\Math{M}(\bold{\Phi},\bold{\Phi},\Pmulti{+2})$), where both massive external states are spin 2, and the massless external state has helicity $+2$ (i.e. the graviton),
\ie
M(\bo^{s=2},\bt^{s=2},P^{+2})=
\frac{\kappa}{2m^2}\Bigg\{&\left[ \alpha_1\alpha_3\left( 1-\frac{\alpha_2}{\alpha_1}\frac{Z_{12}}{2m} \right)\left( 1-\frac{\alpha_4}{\alpha_3}\frac{Z_{34}}{2m} \right)\right] x^2 \LRA{12}^4\\
+&\left[ \alpha_2\alpha_4 \frac{Z_{12}Z_{34}}{2m^2}-\alpha_2\alpha_3\frac{Z_{12}}{2m}-\alpha_1\alpha_4\frac{Z_{34}}{2m} \right] x^2 \LRA{12}^3\frac{\MixLeft{1}{P}{2}}{m}\\
+&\left[ \alpha_2\alpha_4 \frac{Z_{12}Z_{34}}{4m^2} \right] x^2 \LRA{12}^2\left(\frac{\MixLeft{1}{P}{2}}{m}\right)^2\Bigg\}
\;,\fe
and the coupling constants are
\ie
\Rightarrow&(g_0,g_1,g_2)=\\
&\left( \alpha_1\alpha_3 \left( 1-\frac{\alpha_2}{\alpha_1}\frac{Z_{12}}{2m} \right)\left( 1-\frac{\alpha_4}{\alpha_3}\frac{Z_{34}}{2m} \right) , \alpha_2\alpha_4 \frac{Z_{12}Z_{34}}{2m^2}-\alpha_2\alpha_3\frac{Z_{12}}{2m}-\alpha_1\alpha_4\frac{Z_{34}}{2m}  ,\alpha_2\alpha_4 \frac{Z_{12}Z_{34}}{4m^2}\right)
\;.\fe
We can see the component amplitudes contains $g_1$ term, which is forbidden. In order to get rid of the $g_1$-term, we are forced to choose $\alpha_2=\alpha_4=0$\footnote{Neither can we choose $\alpha_1=0$ nor $\alpha_3=0$, because we will get $g_0=0$, and this will violate the equivalence principle, which is not what Einstein would like to see.}. This choice not only excludes the $g_1$ term in $M(\bo^{s=2},\bt^{s=2},P^{+2})$, but $g_1$ term in all component amplitudes that includes graviton, e.g., $M(\bo^{s=1},\bt^{s=1},P^{+2})$. Substitute $\alpha_2=\alpha_4=0$ into \eqref{eqn: N=4 SUGRA alpha beta}, we get
\ie
\label{SUGRA amplitude}
&\Math{M}(\bold{\Phi},\bold{\Phi},\Pmulti{+2})=\frac{\kappa}{2m^2} x^2 S^{(1)}_{12} S^{(1)}_{34} \\
&\Math{M}(\bold{\Phi},\bold{\Phi},\Pmulti{0})=\frac{\kappa}{2m^2}  S^{(2)}_{12}  S^{(2)}_{34}
\;.\fe
In fact, not only $g_1$-terms are excluded in all component amplitudes of this superamplitude, but also $g_2$ terms, leaving with only minimal coupling terms. This indicates that supersymmetry implies all components in massive scalar multiplets interact with graviton only through minimal couplings.


\section{Summary and outlook}

In this paper, we calculated the supersymmetric part in \eqref{M and A}, i.e., $\AAA\left(Z,m\right)$. Starting from $\Math{N}=2$, the delta function is first extracted, then building blocks of the superamplitude, i.e., $\Ha$ and $\Hb$ are examined. There are two solutions in $\Math{N}=2$,, which are related to each other by Grassmann Fourier transformation, see \eqref{S1: the first solution} and \eqref{S2: the second solution}. The most general solution of $\AAA\left(Z,m\right)$ in $\Math{N}=2$ is the linear combination of the two.

Two special cases are discussed, the $Z=0$ case and the BPS case ($Z=2m$). In the $Z=0$ case, the solutions can be factorized into products of $\Math{N}=1$ amplitudes, see \eqref{Z=0 limit MHV amplitude} and \eqref{Z=0 limit MHV_bar amplitude}. In the BPS limit, we know that half of the Grassmannian degrees drop out. It is shown in \eqref{eqn: BPS and non-BPS} that in BPS limit, the amplitude can be factorized into two components that are related by Grassmann Fourier transform, each of which is a BPS amplitude. In other words, non-BPS amplitude is the product of two BPS amplitudes, explaining the drop out of half of the Grassmannian degrees in the BPS limit.

The $\Math{N}=4$ supersymmetry is also considered, and especially super-Maxwell amplitude and SUGRA amplitude. The super-Maxwell amplitude indicates that there is non-minimal coupling of photons, see \eqref{Super-Maxwell amplitude}. The SUGRA amplitude showed that, if we require $g_1=0$, then all non-minimal couplings in graviton exchange vanish as well, see \eqref{SUGRA amplitude}.

Any massive spinning body can be describes as a spinning particle at large distances, and black holes are no exception. For instance, a Kerr black hole can be describe as a massive particle with large spin, interacting with graviton fields through minimal couplings~\cite{Arkani-Hamed:2019ymq}. It would be interesting to see what spinning particles are able to describe a supersymmetric black hole at large distance. Recently, relative entanglement entropy of binary Kerr black holes is found to be nearly zero for minimal coupling in the Eikonal limit, and increases when spin multipole moments are turned on~\cite{Aoude:2020mlg}. We are interested in the relative entanglement entropy of supersymmetric black holes, and examine whether BPS limit results the lowest entropy, relative to non-BPS amplitudes~\cite{Chen:2021huj}.


\section*{Acknowledgements}
I would like to thank Yu-Tin Huang for enlightening discussions whenever I was lost in the calculations. Also thank Ming-Zhi Chung, Man-Kuan Tam for discussions. The work is supported by MoST Grant No. 109-2112-M002-020-MY3.


\appendix

\section{Spinor helicity formalism}

For a more detailed introduction of massive spinor helicity formalism, see~\cite{Arkani-Hamed:2017jhn}.

\paragraph{Contractions and the Levi-Civita Tensor}\label{sec:Levi-Civita}
$ $ \newline
We choose the convention of contracting the dotted and undotted spinors into square and angle brackets as:
\begin{equation}\label{eq:Spinor_Contractions}
\LRA{\lambda \mu} \equiv \lambda^{\alpha} \mu_{\alpha} = \varepsilon_{\alpha \beta} \lambda^\alpha \mu^{\beta}, \quad
\LRB{\lambda \mu} \equiv \tilde{\lambda}_{\dot{\alpha}} \tilde{\mu}^{\dot{\alpha}} = \varepsilon^{\dot{\alpha} \dot{\beta}} \tilde{\lambda}_{\dot{\alpha}} \tilde{\mu}_{\dot{\beta}}
\;.\end{equation}
Same for massive spinors that carry SU(2) indicies. Here the Levi-Civita tensor in matrix form is given by:
\begin{equation}\label{eq:Levi-Civita}
\LCT^{\alpha \beta} = \LCT^{\dot{\alpha} \dot{\beta}} = - \LCT_{\alpha \beta} = - \LCT_{\dot{\alpha} \dot{\beta}} = 
\begin{pmatrix}
0 & 1 \\
-1 & 0
\end{pmatrix}
\;,\end{equation}
such that
\begin{equation}\label{eq:Levi-Civita_contract}
\LCT^{\alpha \beta}\LCT_{\beta \gamma} = \delta^{\alpha}_{\gamma}
\;.\end{equation}

\paragraph{The Massless and Massive Momenta}\label{sec:Momenta_rep}
$ $ \newline
The momentum of the massless particle $P$ can be written as a product of two two-component spinors:
\begin{equation}\label{eq:Massless_momentum_rep}
P_{\alpha \dot{\alpha}} = \lambda_{P\alpha} \tilde{\lambda}_{P\dot{\alpha}} \equiv \RA{\lambda}\LB{\lambda}
\;.\end{equation}
and a massive momentum $p_i$ can be written as a product of two 2-by-2 matrices:
\begin{equation}\label{eq:Massive_momentum_rep}
p_{i\alpha \dot{\alpha}} = \lambda_{i\alpha}^I \tilde{\lambda}_{iI \dot{\alpha}} \equiv \RA{\bi^I}\LB{\bi_{I}}
\;.\end{equation}
where the $I$ is the SU(2) index. 

\paragraph{Explicit kinematics}
$ $\newline
For massless particles with momentum
\ie
p^\mu=
\begin{pmatrix}
E, & E\sin(\theta)\cos(\phi), & E\sin(\theta)\sin(\phi), & E\cos(\theta)
\end{pmatrix}
\;,\fe
we have
\ie
\lambda_\alpha=\sqrt{2E}
\begin{pmatrix}
s\\
-c
\end{pmatrix}
\;,\;
\tilde\lambda_{\dot\alpha}=\sqrt{2E}
\begin{pmatrix}
s^*\\
-c
\end{pmatrix}
\;.\fe
\\
For massive particles with momentum
\ie
p^\mu=
\begin{pmatrix}
E, & p\sin(\theta)\cos(\phi), & p\sin(\theta)\sin(\phi), & p\cos(\theta)
\end{pmatrix}
\;,\fe
we have
\ie
\lambda_\alpha^I=
\begin{pmatrix}
&\sqrt{E+p}s & \sqrt{E-p}c^* \\
&-\sqrt{E+p}c & \sqrt{E-p}s
\end{pmatrix}
\;,\;
\tilde\lambda_{\dot\alpha}^I=
\begin{pmatrix}
&\sqrt{E-p}c & \sqrt{E-p}s \\
&-\sqrt{E+p}s & \sqrt{E+p}c^*
\end{pmatrix}
\;,\fe
where $c\equiv\cos\left(\frac{\theta}{2}\right)e^{i\phi}$, $s\equiv\sin\left(\frac{\theta}{2}\right)$.

\paragraph{Contractions of Massive Spinors}
$ $ \newline
The on-shell condition for massive particle is given by 
\begin{equation}\label{eq:Massive-on-shell-condition}
p_i^2 = \LRA{\bi^I \bi^J}\LRB{\bi_J \bi_I} = m_i^2
\;,\end{equation}
where we choose
\begin{equation}\label{eq:Ang_and_sqr_same_massive_spinor_fixed}
\boxed{\LRA{\bi^I \bi^J} = -m_i \LCT^{IJ},\quad \LRB{\bi^I \bi^J} = +m_i \LCT^{IJ}.}
\end{equation}
which will be repeatedly used throughout massive amplitude calculations
With these conventions, we find that momentum contracting with the massive spinors are:
\ie
&m_i \RA{\bi^I} = +p_i \RB{\bi^I}, \quad m_i\RB{\bi^I} = +p_i \RA{\bi^I} \\
&m_i \LA{\bi^I} = - \LB{\bi^I}p_i, \quad m_i\LB{\bi^I} = - \LA{\bi^I}p_i 
\;.\fe
For massive spinors associated with the same particle whose LG indices contracted:
\ie
&\RA{\bi^I}\epsilon_{IJ}\LA{\bi^J}=-m_{i}\epsilon^{\alpha\dot{\alpha}}\;;\;
\RB{\bi^I}\epsilon_{IJ}\LB{\bi^J}=m_{i}\epsilon^{\alpha\dot{\alpha}}
\\
&\RA{\bi^I}\epsilon_{IJ}\LB{\bi^J}=p_{i\alpha\dot{\alpha}}\;;\;
\RB{\bi^I}\epsilon_{IJ}\LA{\bi^J}=-p_{i}^{\alpha\dot{\alpha}}
\;.\fe

\paragraph{Definition and convention of the $x$-factor}
$ $ \newline
The external momenta satisfy
\begin{equation}\label{eq:3pt_momentum_conservation}
p_1 + p_2 + P = 0
\;.\end{equation}
The momentum conservation condition \eqref{eq:3pt_momentum_conservation} and the on-shell condition yields:
\begin{equation}\label{eq:3pt_momentum_conservation_in_bra-ket}
2p_1 \cdot P = \MixLeft{P}{p_1}{P} = \lambda^{\alpha}_P p_{1 \alpha \dot{\alpha}}\tilde{\lambda}_P^{\dot{\alpha}} = 0
\;,\end{equation}
so that $\lambda_{P\alpha}$ is proportional to $ p_{1 \alpha \dot{\alpha}}\tilde{\lambda}_P^{\dot{\alpha}} $. This allow us to define the $x$-factor:
\begin{equation}\label{eq:x-def}
\boxed{
x \lambda_{P \alpha} = \frac{p_{1 \alpha \dot{\alpha}}\tilde{\lambda}_P^{\dot{\alpha}}}{m}, \quad \frac{\tilde{\lambda}_P^{ \dot{\alpha}} }{x} = \frac{p_1 ^{\dot{\alpha} \alpha} \lambda_{P \alpha}}{m} 
}
\end{equation}


\section{Grassmann variables}
\label{appendix: Grassmannian variables}

\paragraph{Convention of the Grassmann variables $\eta$}
$ $ \newline
When contract with Levi-Civita tensors, the massive Grassmann variables $\eta$'s transform as
\ie
\eta_{iIA}\equiv\epsilon_{AB}\eta_{iI}^B &\Rightarrow \frac{\partial}{\partial\eta_{iIA}}=-\epsilon^{AB}\frac{\partial}{\partial\eta_{iI}^B}
\\
\eta_i^{IA}\equiv\epsilon^{IJ}\eta_{iJ}^A &\Rightarrow\frac{\partial}{\partial\eta_i^{IA}}=-\epsilon_{IJ}\frac{\partial}{\partial\eta_{iJ}^A}
\;.\fe
To simplify equations, let's denote the contraction rules for massive Grassmann variables $\eta_i$ and for massless Grassmann variables $\eta_P$
\ie
\label{eqn: appendix: contraction rules}
&\eta_i^A\cdot\eta_i^B\equiv-\frac{1}{2}\epsilon^{IJ}\eta_{iI}^A\eta_{iJ}^B\\
&\eta_P\cdot\eta_P\equiv\frac{1}{2}\eta_P^A\eta_{PA}
\;.\fe

\paragraph{Grassmann Fourier transformation}
$ $ \newline
The Grassmann Fourier transformation of a function $f(\eta)$ in $\eta$ basis to $\bar\eta$ basis is
\ie
&\bar f(\bar\eta) =\int d\eta\cdot e^{\bar\eta\eta}f(\eta)\\
&f(\eta) =\int d\bar\eta\cdot e^{-\eta\bar\eta}\bar f(\bar\eta)
\;.\fe
For Grassmann variables of massless particles, the Grassmann Fourier transformation is
\ie
\label{eqn: fourier transform-massless}
\bar f(\bar\eta) =\int d\eta^A d\eta_A\cdot e^{\bar\eta^A\eta_A}f(\eta)
\;,\fe
while for Grassmann variables of massive particles
\ie
\label{eqn: fourier transform-massive}
\bar f(\bar\eta) =\int d\eta_I^A d\eta^{IB}d\eta_{JA}d\eta_{B}^J\cdot e^{\bar\eta_I^A\eta_A^I}f(\eta)
\;.\fe


\section{More on the building blocks}
\label{appendix: More on Ha and Hb}

$\Ha^A$ and $\Hb_A$ play crucial rules in this paper, they serve as basic building blocks of the superamplitude. This appendix is devoted to introducing several important properties of them, which will be useful if one wants to reproduce the calculations we have done in this paper. To simplify equations, let's denote
\ie
\etas{i}{A}\equiv-{1\over2}\epsilon^{IJ}\eta_{iI}^A\eta_{iJ}^A \;\;\text{(no sum over $i$ and $A$)} \Rightarrow \eta_{iI}^A\eta_{iJ}^A=\epsilon_{IJ}\etas{i}{A}
\;,\fe
which is just a special case of \eqref{eqn: appendix: contraction rules}.

\begin{lemma}
\label{appendix lemma: LRA{P1}eta_{1}^A}
\ie
\LRA{P\bo^I}\eta_{1I}^A \eqd \frac{1}{x}\Ha^A-\Hb^A
\;.\fe
\end{lemma}
\begin{proof}
\ie
&\LRB{\zeta P}\left( {1\over x}\Ha^A-\Hb^A \right)=\MixRight{\zeta}{p_1}{P}\Ha^A-\LRB{\zeta P}\Hb^A\\
\eqd &\sum_{i=1}^2\MixRight{\zeta}{p_1}{\bi^I}\eta_{iI}^A-\sum_{i=1}^2\LRB{\zeta \bi^I}\sigma_i\eta_{iI}^A\\
=&\MixRight{\zeta}{(p_1+p_2)}{\bt^I}\eta_{2I}^A \eqd  \LRB{\zeta P}\LRA{P\bo^I}\eta_{1I}^A
\;.\fe
\end{proof}

\begin{lemma}
\ie
&\LRA{\bo^I\bt^J}\eta_{1I}^A\eta_{2JA} \eqd -\frac{2}{x}\HHa+2\HHab\\
&\LRB{\bo^I\bt^J}\eta_{1I}^A\eta_{2JA}\eqd -2x\HHb-2\HHab
\;.\fe
\end{lemma}
\begin{proof}
$ $ \newline
Let's just proof the first equation in the lemma, since the proof of the second is similar to that of the first.
\ie
&\LRA{\bo^I\bt^J}\eta_{1I}^A\eta_{2JA}\\
=&\LRA{\bo^I\bo^J}\eta_{1I}^A\eta_{1JA} +\LRA{\bo^I\bt^J}\eta_{1I}^A\eta_{2JA} \\
\eqd&\LRA{\bo^I P}\eta_{1I}^A \Ha_A \eqd -\frac{2}{x}\HHa+2\HHab
\;,\fe
where the last equality follows from Lemma \ref{appendix lemma: LRA{P1}eta_{1}^A}.
\end{proof}

\begin{lemma}
\ie
\label{eqn: delta H_a H_b}
\begin{cases}
\delta^1\Ha^1\eqd -{1\over 2}\sum_{i,j}\LRA{\bi^I\bj^J}\eta_{iI}^1\eta_{jJ}^1\\
\delta^2\Ha^2\eqd-{1\over 2}\sum_{i,j}\LRA{\bi^I\bj^J}\eta_{iI}^2\eta_{jJ}^2\\
\delta^1\Hb_2\eqd{1\over 2x}\sum_{i,j}\LRB{\bi^I\bj^J}\eta_{iI}^1\eta_{jJ}^1\\
\delta^2\Hb_1\eqd-{1\over 2x}\sum_{i,j}\LRB{\bi^I\bj^J}\eta_{iI}^2\eta_{jJ}^2
\end{cases}
\;.\fe
\end{lemma}
\begin{proof}
$ $ \newline
\ie
\delta^1\Ha^1=&\delta^1\cdot \frac{1}{\LRA{\zeta P}} \sum_{j=1}^2\LRA{\zeta \bj^J}\eta_{jJ}^A\\
=&\frac{1}{\LRA{\zeta P}}\sum_{i,j}\LRA{P \bi^I}\LRA{\zeta \bj^J} \eta_{iI}^A\eta_{jJ}^A\\
=&-{1\over2}\frac{1}{\LRA{\zeta P}}\sum_{i,j}\LRA{\zeta P}\LRA{\bi^I\bj^J} \eta_{iI}^A\eta_{jJ}^A
\;,\fe
where we used Schouten identity. The other equations can be proven in similar ways.
\end{proof}

\begin{lemma}
\ie
\label{eqn: delta H_a H_b = R}
\begin{cases}
\delta^1\Ha^1\Hb_2\eqd-{1\over x}\left( \LRB{P\bo^I}\eta_{1I}^1\etas{2}{1}+\LRB{P\bt^I}\eta_{2I}^1\etas{1}{1} \right)\\
\delta^2\Ha^2\Hb_1\eqd{1\over x}\left( \LRB{P\bo^I}\eta_{1I}^2\etas{2}{2}+\LRB{P\bt^I}\eta_{2I}^2\etas{1}{2} \right)
\end{cases}
\;.\fe
\end{lemma}

\begin{proof} Let $m=1$,
\ie
&\LRA{\zeta P}\delta^1\Ha^1\Hb_2\\
=&-\left( \LRA{\zeta P}\Ha^1 \right)\left( \delta^1\Hb_2 \right)\\
=&{1\over x}\left( \LRA{\zeta\bo^I}\eta_{1I}^1+\LRA{\zeta\bt^I}\eta_{2I}^2 \right)\left( \LRB{\bo^K\bt^L}\eta_{1K}^1\eta_{2L}^1+\etas{1}{1}+\etas{2}{1} \right)\\
=&{1\over x}\left( \MixLeft{\zeta}{p_1}{\bt^I}\eta_{2I}^1\etas{1}{1}+\LRA{\zeta\bo^I}\eta_{1I}^1\etas{2}{1}+\MixLeft{\zeta}{p_2}{\bo^I}\eta_{1I}^1\etas{2}{1}+\LRA{\zeta\bt^I}\eta_{2I}^1\etas{1}{1} \right)\\
=&-\LRA{\zeta P}\cdot{1\over x} \left( \LRB{P\bo^I}\eta_{1I}^1\etas{2}{1}+\LRB{P\bt^I}\eta_{2I}^1\etas{1}{1} \right)
\;.\fe
The second equation can be proven in similar fashion.
\end{proof}

When $\Ha^A$, $\Hb_A$, and their products meet the delta function, we are able to recast them into a Lorentz invariant and R-charge symmetric form. To simplify notations, let's define

\ie
&\Math{Q}_{\alpha}^A\equiv\lambda_{1\alpha}^I\eta_{1I}^A+\lambda_{2\alpha}^I\eta_{2I}^A\;\;;\;\;
\tilde{\Math{Q}}_{\dot{\alpha}}^A\equiv\tilde{\lambda}_{1\dot{\alpha}}^I\eta_{1I}^A-\tilde{\lambda}_{2\dot{\alpha}}^I\eta_{2I}^A\\
&\Math{Q}_{P\alpha}^A\equiv\lambda_{P\alpha}\eta_{P}^A\;\;;\;\;
\tilde{\Math{Q}}_{P\dot{\alpha}}^A\equiv\tilde{\lambda}_{P\dot{\alpha}}\eta_{P}^A\\
&\eta_i^A\cdot\eta_i^B\equiv-\frac{1}{2}\epsilon^{IJ}\eta_{iI}^A\eta_{iJ}^B\;\;;\;\;\eta_P\cdot\eta_P\equiv\frac{1}{2}\eta_P^A\eta_{PA}
\;.\fe
The following equations follow from \eqref{eqn: delta H_a H_b}
\ie
\label{eqn: meeting the delta functions (order 4)}
&\delta^1\delta^2=-\frac{1}{2}\LRA{\Math{Q}^AP}\LRA{\Math{Q}_AP}
\\
&\delta^1\delta^2\left(\eta_P\cdot\eta_P\right)=\frac{1}{4}\LRA{\Math{Q}^A\Math{Q}_P^{B}}\LRA{\Math{Q}_A\Math{Q}_{PB}}
\\
&\delta^1\delta^2\HHa=\frac{1}{12}\LRA{\Math{Q}^A\Math{Q}^B}\LRA{\Math{Q}_A\Math{Q}_B}
\\
&\delta^1\delta^2\HHab=-\frac{1}{12x}\LRA{\Math{Q}^A\Math{Q}^B}\LRB{\tilde{\Math{Q}}_A\tilde{\Math{Q}}_B}
\\
&\delta^1\delta^2\HHb=-\frac{1}{12x^2}\LRB{\tilde{\Math{Q}}^A\tilde{\Math{Q}}^B}\LRB{\tilde{\Math{Q}}_A\tilde{\Math{Q}}_B}
\\
&\delta^1\delta^2\Ha^A\eta_{PA}=\frac{1}{3}\LRA{\Math{Q}^A\Math{Q}^B}\LRA{\Math{Q}_A\Math{Q}_{PB}}
\\
&\delta^1\delta^2\eta_{P}^A\Hb_A=-\frac{1}{3x}\LRB{\tilde{\Math{Q}}^A\tilde{\Math{Q}}^B}\LRA{\Math{Q}_A\Math{Q}_{PB}}
\;.\fe
Things will be more complicated when $\delta^1\Ha^1\Hb_2$ or $\delta^2\Ha^2\Hb_1$ appear, but as we can see from \eqref{eqn: delta H_a H_b = R}, we have
\ie
\label{eqn: meeting the delta functions (order 5)}
&\delta^1\delta^2\HHa\eta_P^A\Hb_A=-\frac{2}{9x}\LRA{\Math{Q}^{A} \Math{Q}^{B}}  \LRB{\tilde{\Math{Q}}_{A} \tilde{\Math{Q}}_{P}^{C}}\left(\ETAsL{B}{C}\right)
\\
&\delta^1\delta^2\HHb\Ha^A\eta_{PA}=\frac{2}{9x^2}\LRB{\tilde{\Math{Q}}^A \tilde{\Math{Q}}^{B}} \LRB{\tilde{\Math{Q}}_A\tilde{\Math{Q}}_P^{C}} \left( \ETAsL{B}{C} \right)
\;,\fe
and finally
\ie
\label{eqn: meeting the delta functions (order 6)}
&\delta^1\delta^2\HHa\HHb\\
&\;\;\;\;\;=-\frac{1}{6x^2}\LRB{\tilde{\Math{Q}}^AP} \LRB{\tilde{\Math{Q}}_AP} \left( \ETAsU{B}{C} \right) \left( \ETAsL{B}{C} \right)
\\
&\delta^1\delta^2\HHa\HHb\left(\eta_P\cdot\eta_P\right)\\
&\;\;\;\;\;=\frac{1}{12x^2} \LRB{\tilde{\Math{Q}}^{(A} \tilde{\Math{Q}}_{P}^{B)}} \LRB{\tilde{\Math{Q}}_{A} \tilde{\Math{Q}}_{PB}} \left(\ETAsU{C}{D} \right)\left(\ETAsL{C}{D} \right)
\;.\fe


\section{More on $\MmmPr$}
\label{appendix: amplitude = f(Ha,Hb)}

The amplitude is proportional to $\delta^1\delta^2$ follows straightforwardly given the generators, see \eqref{eqn: A=delta B}. The main goal in this section is to study the rest of the amplitude, i.e., $\MmmPr$, in more detail, and proof that it is a function of $\Ha$ and $\Hb$, in other words, all $\eta_i$ ($\eta_i=\;\eta_1\text{ or }\eta_2$) in $\MmmPr$ must be of the form
\ie
\sum_{i=1}^2\lambda_i ^I\eta_{iI}^A  \;\;\text{or}\;\;  \sum_{i=1}^2\tilde\lambda_i ^I\sigma_i\eta_{iI}^A
\;.\fe
$\MmmPr$ is a function of $\eta_i$'s, and we can expand it according to the order of $\eta_i$'s
\ie
\MmmPr=\MmmPr^{(0)}+\MmmPr^{(1)}+\MmmPr^{(2)}+\cdots
\;.\fe
Observe the form of the generators \eqref{eqn: QQ+ full}, they have a multiplicative part in $\eta_i$ and a differential part in $\eta_i$. Therefore, since $\MmmPr$ must satisfy \eqref{eqn:constraints on B}, we have (see \eqref{eqn: definition of D and D+} for definitions)
\ie
\label{appendix eqn: recursion of Mres}
&\DaRaise{A} \MmmPr^{(n-2)} + \DaLower{A}\MmmPr^{(n)} + \LRA{\zeta P}\eta_P^A\MmmPr^{(n-1)} \eqd0\\
&\DbRaise{A} \MmmPr^{(n-2)} + \DbLower{A}\MmmPr^{(n)} + \LRB{\xi P}\frac{\p}{\p\eta_P^A}\MmmPr^{(n-1)} \eqd0
\;.\fe

\begin{lemma}
$ $\newline
The lowest order of $\MmmPr$ must be $\MmmPr^{(0)}$.
\end{lemma}
\begin{proof}
$ $\newline
Assume the lowest order of $\MmmPr$ is $\MmmPr^{(low)}$, and we must have
\ie
\label{appendix eqn: lowest Mres-satisfy}
&\DaLower{A}\MmmPr^{(low)} \eqd 0 \Rightarrow
\sum_{i=1}^2 \LRA{\zeta \bi^I} \sigma_i\frac{\partial}{\partial\eta_{iA}^I}\MmmPr^{(low)}\eqd0 \\
&\DbLower{A}\MmmPr^{(low)} \eqd 0 \Rightarrow
\sum_{i=1}^2 \LRB{\xi \bi^I}  \frac{\partial}{\partial\eta_{i}^{IA}}\MmmPr^{(low)}\eqd0
\;.\fe
Since $\frac{\p}{\eta_{i}^{IA}}\MmmPr^{(n)}$ carries a LG index $I$, we have three possibilities for $\frac{\p}{\eta_{i}^{IA}}\MmmPr^{(n)}$
\ie
\label{appendix eqn: lowest Mres-diff}
&\frac{\p}{\eta_{i}^{IA}}\MmmPr^{(low)}=\eta_{iIA}\;\;\text{or}\;\;\LRA{\bi_I f_{iA}}\;\;\text{or}\;\; \LRB{\bi_I \tilde f_{iA}}
\;.\fe
However, the first one is impossible since it would imply $\MmmPr^{(low)}\propto\eta_i^{IA}\eta_{iIA}=0$. In addition, one of the second and the third is redundant, since we can always use $p_i\RA{\bi^I}=\RB{\bi^I}$ to convert $\RA{\bi^I}$ to $\RB{\bi^I}$. Let's keep $\frac{\p}{\eta_{i}^{IA}}\MmmPr^{(low)}=\LRA{\bi_I f_{iA}}$\footnote{Note that $\RA{f_{iA}}$ might carry Grassmann variables.} and \eqref{appendix eqn: lowest Mres-satisfy} implies
\begin{subequations}
\begin{align}
&\LRA{\zeta f_{1A}}-\LRA{\zeta f_{2A}} \eqd 0 \label{appendix eqn: lowest=0 1}\\
&\MixRight{\xi}{p_1}{f_{1A}}+\MixRight{\xi}{p_2}{f_{2A}} \eqd 0 \label{appendix eqn: lowest=0 2}
\;.\end{align}
\end{subequations}
\eqref{appendix eqn: lowest=0 1} implies $\RA{f_{1A}}=\RA{f_{2A}}\equiv \RA{f_A}$. Set $\LA{\zeta}=\LA{\xi}p_2$, and add \eqref{appendix eqn: lowest=0 1} to \eqref{appendix eqn: lowest=0 2}, we get
\ie
&\MixRight{\xi}{(p_1+p_2)}{f_{A}}=-\LRB{\zeta P}\LRA{P f_{A}}\eqd0\Rightarrow\LRA{P f_{A}}\eqd0\Rightarrow\RA{f_{A}}\eqd f_{A}\RA{P}
\;.\fe
Since $\frac{\p}{\eta_{i}^{IA}}\MmmPr^{(low)}=\LRA{\bi_I f_{A}}$, $\RA{f_A}$ can not be a function of $\eta_{1I}^A$ and $\eta_{2I}^A$, and the solution to $\frac{\p}{\eta_{i}^{IA}}\MmmPr^{(low)}\eqd f_{A}\LRA{\bi_I P}$ is
\ie
\MmmPr^{(low)}= f_{A} \sum_i \eta_i^{IA}\LRA{\bi_I P}+\text{($\eta_i$-free term)} \eqd \text{($\eta_i$-free term)}
\;,\fe
where in the last equation, we used $\sum_i \eta_i^{IA}\LRA{\bi_I P}\eqd0$. Therefore, $\MmmPr^{(low)}$ is $\eta_i$-free, in other words, $\MmmPr^{(low)}=\MmmPr^{(0)}$.
\end{proof}

\begin{lemma}
\label{lemma: Recursion relation}
$ $\newline
If $\MmmPr^{(n-1)}$ and $\MmmPr^{(n-2)}$ is a function of $\Ha$ and $\Hb$, then so is $\MmmPr^{(n)}$.
\end{lemma}
\begin{proof}
$ $\newline
Define $\frac{\p}{\eta_{i}^{IA}}\MmmPr^{(n)}\equiv\LRA{i_I f_{iA}^{(n)}}$ and \eqref{appendix eqn: recursion of Mres} implies
\begin{subequations}
\begin{align}
\SinT\sum_i m \sigma_i\LRA{\zeta f_{i}^{(n)A}}& \eqd
\LRA{\zeta P} \left( \CosT\Ha^A\MmmPr^{(n-2)}+\eta_P^A\MmmPr^{(n-1)} \right) \label{appendix eqn: M(n-2) and f(n) 1}\\
\CosT\sum_i \MixRight{\xi}{p_i}{f_{i}^{(n)A}}& \eqd
\LRB{\xi P} \left( -\SinT\Hb^A\MmmPr^{(n-2)}+\frac{\p}{\p\eta_P^A}\MmmPr^{(n-1)} \right) \label{appendix eqn: M(n-2) and f(n) 2}
\;.\end{align}
\end{subequations}
If we set $\LA{\zeta}=\LB{\xi}p_i$ and linear combine \eqref{appendix eqn: M(n-2) and f(n) 1} and \eqref{appendix eqn: M(n-2) and f(n) 2}, we get
\ie
\CosT\SinT\LRA{P f_{1}^{(n)A}} \eqd -\left( \frac{\CosTsqr}{x}\Ha^A+\SinTsqr\Hb^A \right)\MmmPr^{(n-2)}
-\CosT\eta_P^A\MmmPr^{(n-1)}+\SinT\frac{\p}{\p\eta_P^A}\MmmPr^{(n-1)}
&\;\;\text{,if $\LA{\zeta}=\LB{\xi}p_2$}\\
\CosT\SinT\LRA{P f_{2}^{(n)A}} \eqd -\left( \frac{\CosTsqr}{x}\Ha^A+\SinTsqr\Hb^A \right)\MmmPr^{(n-2)}
-\CosT\eta_P^A\MmmPr^{(n-1)}+\SinT\frac{\p}{\p\eta_P^A}\MmmPr^{(n-1)}
&\;\;\text{,if $\LA{\zeta}=\LB{\xi}p_1$}
\;.\fe
This not only tells us that both $f_{1}^{(n)A}$ and $f_{2}^{(n)A}$ are functions of $\Ha$ and $\Hb$, but also implies
\ie
\label{appendix eqn: f1 and f2}
\LRA{P f_{1}^{(n)A}}=\LRA{P f_{2}^{(n)A}}
\;.\fe
This implies whenever $\lambda_{i\alpha}^I\eta_{iI}^A$ appears in $\MmmPr^{(n)}$, it must be either $\LRA{\zeta\bo^I}\eta_{1I}^A+\LRA{\zeta\bt^I}\eta_{2I}^A$ or $\LRA{P\bo^I}\eta_{1I}^A \eqd -\LRA{P\bt^I}\eta_{2I}^A$. In other words, $\LRA{\zeta\bo^I}\eta_{1I}^A-\LRA{\zeta\bt^I}\eta_{2I}^A$ can't appear in $\MmmPr$. In addition, $\LRB{\xi\bo^I}\eta_{1I}^A-\LRB{\xi\bt^I}\eta_{2I}^A$ is also a possible choice, since
\ie
\LRB{\xi\bo^I}\eta_{1I}^A-\LRB{\xi\bt^I}\eta_{2I}^A=-\MixRight{\xi}{p_2}{\bo^I}\eta_{1I}^A-\MixRight{\xi}{p_2}{\bt^I}\eta_{1I}^A+\LRB{\xi P}\LRA{P\bo^I}\eta_{1I}^A
\;.\fe
By similar arguments, $\LRB{\xi\bo^I}\eta_{1I}^A+\LRB{\xi\bt^I}\eta_{2I}^A$ is prohibited. Therefore, only the following terms can exist in $\MmmPr^{(n)}$,
\ie
&\LRA{\zeta\bo^I}\eta_{1I}^A+\LRA{\zeta\bt^I}\eta_{2I}^A=\LRA{\zeta P}\Ha^A \\
&\LRB{\xi\bo^I}\eta_{1I}^A-\LRB{\xi\bt^I}\eta_{2I}^A=\LRB{\xi P}\Hb^A \\
&\LRA{P\bo^I}\eta_{1I}^A\eqd \frac{1}{x}\Ha^A-\Hb^A
\;.\fe
The last identity follows from Lemma \ref{appendix lemma: LRA{P1}eta_{1}^A}. As a result, $\MmmPr^{(n)}$ is a function of $\Ha$ and $\Hb$.
\end{proof}

Given Lemma \ref{lemma: Recursion relation}, and given $\MmmPr^{(-1)}=0$, we can conclude that $\MmmPr^{(1)}$ is function of $\Ha$ and $\Hb$. By iteration, we can conclude all $\MmmPr^{(n)}$, and thus $\MmmPr$, are functions of $\Ha$ and $\Hb$.

\bibliography{mybib}{}
\bibliographystyle{unsrt}

\end{document}

%% file: latex_command.tex
\def\ie{\begin{equation}\begin{aligned}}
\def\fe{\end{aligned}\end{equation}}

\newtheorem{theorem}{Theorem}[section]
\newtheorem{lemma}[theorem]{Lemma}

\def\p{\partial}

\newcommand{\LRA}[1]{\langle #1  \rangle}
\newcommand{\LRB}[1]{[ #1 ]}
\newcommand{\MixLeft}[3]{\langle #1 | #2 | #3 ]}
\newcommand{\MixRight}[3]{[ #1 | #2 | #3 \rangle}
\newcommand{\RA}[1]{| #1 \rangle}
\newcommand{\LA}[1]{\langle #1 |}
\newcommand{\RB}[1]{| #1 ]}
\newcommand{\LB}[1]{[ #1 |}

\newcommand{\bo}{\bold{1}}
\newcommand{\bt}{\bold{2}}
\newcommand{\bi}{\bold{i}}
\newcommand{\bj}{\bold{j}}

\newcommand{\etas}[2]{(\eta_{#1}^{#2}\cdot\eta_{#1}^{#2})}

\newcommand{\Math}[1]{\mathcal{ #1}}

\newcommand{\ETAsU}[2]{(\eta_1^{#1}\cdot\eta_1^{#2})-(\eta_2^{#1}\cdot\eta_2^{#2})}
\newcommand{\ETAsL}[2]{(\eta_{1 #1}\cdot\eta_{1 #2})-(\eta_{2 #1}\cdot\eta_{2 #2})}

\newcommand{\LCT}{\epsilon}

\def\Ha{\mathcal{H}}
\def\Hb{\bar{\mathcal{H}}}
\def\HHa{(\Ha\cdot\Ha)}
\def\HHb{(\Hb\cdot\Hb)}
\def\HHab{(\Ha\cdot\Hb)}

\newcommand{\DaRaise}[1]{\LRA{\zeta\mathcal{D}^{#1}}_+}
\newcommand{\DbRaise}[1]{\LRB{\xi\tilde{\mathcal{D}}_{#1}}_+}
\newcommand{\DaLower}[1]{\LRA{\zeta\mathcal{D}^{#1}}_-}
\newcommand{\DbLower}[1]{\LRB{\xi\tilde{\mathcal{D}}_{#1}}_-}
\newcommand{\Da}[1]{\LRA{\zeta\mathcal{D}^{#1}}}
\newcommand{\Db}[1]{\LRB{\xi\mathcal{D}_{#1}}}

\newcommand{\tZ}{\theta_Z}
\newcommand{\tZa}{\theta_{Z_{12}}}
\newcommand{\tZb}{\theta_{Z_{34}}}
\newcommand{\CosT}{c_z}
\newcommand{\SinT}{s_z}
\newcommand{\CosTa}{\cos(\tZa)}
\newcommand{\SinTa}{\sin(\tZa)}
\newcommand{\CosTb}{\cos(\tZb)}
\newcommand{\SinTb}{\sin(\tZb)}
\newcommand{\CosTsqr}{c_z^2}
\newcommand{\SinTsqr}{s_z^2}

\newcommand{\Pmulti}[1]{\Math{P}^{#1}}

\newcommand{\MmmP}{\Math{M}({\bo},{\bt},\Pmulti{})}
\newcommand{\MmmPr}{\Math{M}_{res}}
\newcommand{\eqd}{\stackrel{\delta}{=}}

\newcommand{\AAA}{\Math{A}}
\newcommand{\AAAr}{\Math{A}_{res}}
\newcommand{\Ah}{\bold{h}}
\newcommand{\Ag}{\bold{g}}
\newcommand{\Af}{\bold{f}}

\newcommand{\field}[1]{\left|#1\right\rangle}

\newcommand{\Msasb}[3]{{M}({\bo}^{(I_1\cdots I_{2{#1}})},{\bt}^{(J_1\cdots J_{2{#2}})},P^{#3})}
\newcommand{\barMsasb}[3]{\bar{M}({\bo}^{(I_1\cdots I_{2{#1}})},{\bt}^{(J_1\cdots J_{2{#2}})},P^{#3})}

\newcommand{\sMsasb}[3]{\Math{M}_Z({\bo}^{(I_1\cdots I_{2{#1}})},{\bt}^{(J_1\cdots J_{2{#2}})},\Pmulti{#3})}

\newcommand{\sMZaZbsasb}[3]{\Math{M}_{Z_{12},Z_{34}}({\bo}^{(I_1\cdots I_{2{#1}})},{\bt}^{(J_1\cdots J_{2{#2}})},\Pmulti{#3})}

\newcommand{\tQ}{\tilde Q}

%% file: authorlist.tex

\def\AddAuthor#1#2#3#4{%
	\def\Name{#1}%
	\def\PriAf{#2}%
	\def\SecAf{#3}%
	\def\ExtAf{#4}%
	\def\empty{}%
	
	\ifx\PriAf\empty
		\author{#1\fnref{#4}}%
	\else
		\ifx\SecAf\empty
			\ifx\ExtAf\empty
				\author[#2]{#1}%
			\else
				\author[#2]{#1\fnref{#4}}%
			\fi
		\else
			\ifx\ExtAf\empty
				\author[#2,#3]{#1}%
			\else
				\author[#2,#3]{#1\fnref{#4}}%
			\fi
		\fi
	\fi
}


\def\AddCorrespondingAuthor#1#2#3#4#5#6{%
	\def\Name{#1}%
	\def\PriAf{#2}%
	\def\SecAf{#3}%
	\def\ExtAf{#4}%
	\def\empty{}%
	
	\ifx\PriAf\empty
		\author{#1\fnref{#4}\corref{CA}}%
	\else
		\ifx\SecAf\empty
			\ifx\ExtAf\empty
				\author[#2]{#1\corref{CA}}%
			\else
				\author[#2]{#1\fnref{#4}\corref{CA}}%
			\fi
		\else
			\ifx\ExtAf\empty
				\author[#2,#3]{#1\corref{CA}}%
			\else
				\author[#2,#3]{#1\fnref{#4}\corref{CA}}%
			\fi
		\fi
	\fi

	\ead{#6}%
	\cortext[CA]{#5}%
}


\def\AddInstitute#1#2{%
	\address[#1]{#2}%
}


\def\AddExternalInstitute#1#2{%
	\fntext[#1]{#2}%
}
